\documentclass{article}

\usepackage{microtype}
\usepackage{graphicx}
\usepackage{subfigure}
\usepackage{booktabs} 
\usepackage[normalem]{ulem}

\usepackage{tikz}
\usepackage{tikzit}

\tikzstyle{basic}=[fill=white, draw=black, shape=circle]
\tikzstyle{square}=[fill=white, draw=black, shape=rectangle]
\tikzstyle{big dashed}=[fill=white, draw=black, shape=circle, minimum width=1cm, dashed]
\tikzstyle{vertical ellipse dashed}=[fill=none, draw=blue, minimum width=0.75cm, minimum height=3cm, ellipse, dashed, tikzit shape=rectangle, tikzit draw=blue, tikzit fill=white]
\tikzstyle{small vertical ellipse dashed}=[fill=none, draw=blue, shape=circle, tikzit fill=white, tikzit draw=blue, dashed, minimum width=0.75cm, minimum height=1.5cm, tikzit shape=rectangle, ellipse]
\tikzstyle{tiny vertical ellipse dashed}=[fill=none, draw=blue, shape=circle, tikzit fill=white, ellipse, dashed, minimum width=0.75cm, minimum height=1cm, tikzit shape=rectangle]
\tikzstyle{red}=[fill=red, draw=black, shape=circle]
\tikzstyle{green}=[fill={rgb,255: red,0; green,128; blue,128}, draw=black, shape=circle]
\tikzstyle{blue}=[fill=blue, draw=black, shape=circle]
\tikzstyle{huge dashed}=[fill=white, draw=black, shape=circle, dashed, minimum width=2cm]
\tikzstyle{medium}=[fill=white, draw=black, shape=circle, minimum width=1cm]
\tikzstyle{pale green}=[fill={rgb,255: red,173; green,231; blue,0}, draw=black, shape=circle, minimum width=1cm]
\tikzstyle{horizontal ellipse dashed}=[fill=white, draw=black, tikzit draw=magenta, tikzit shape=rectangle, minimum width=3cm, minimum height=0.75cm, ellipse, dashed]
\tikzstyle{minsize}=[fill=white, draw=black, shape=circle, minimum width=0.75cm]
\tikzstyle{horizontal ellipse green}=[fill={rgb,255: red,191; green,255; blue,0}, draw=black, tikzit draw={rgb,255: red,191; green,255; blue,0}, tikzit shape=rectangle, minimum width=3cm, minimum height=0.75cm, ellipse, dashed]
\tikzstyle{horizontal ellipse blue}=[fill={rgb,255: red,107; green,203; blue,255}, draw=black, tikzit draw=blue, tikzit shape=rectangle, minimum width=3cm, minimum height=0.75cm, ellipse, dashed]
\tikzstyle{smallblack}=[fill=black, draw=black, shape=circle, inner sep=0 pt, minimum size=5 pt]
\tikzstyle{smallSquare}=[fill=white, draw=black, shape=rectangle, inner sep=0 pt, minimum size=6 pt]
\tikzstyle{smallCircle}=[fill=white, draw=black, shape=circle, inner sep=0 pt, minimum size=20 pt]
\tikzstyle{big vertical ellipse dashed}=[fill=none, draw=blue, shape=circle, tikzit shape=rectangle, ellipse, dashed, minimum width=0.95cm, minimum height=3.7cm]
\tikzstyle{smallred}=[fill={rgb,255: red,222; green,45; blue,38}, draw={rgb,255: red,222; green,45; blue,38}, shape=circle, inner sep=0 pt, minimum size=5 pt]
\tikzstyle{smallblue}=[fill=blue, draw=blue, shape=circle, inner sep=0pt, minimum size=5pt]
\tikzstyle{small green}=[fill={rgb,255: red,0; green,107; blue,61}, draw={rgb,255: red,0; green,107; blue,61}, shape=circle, inner sep=0pt, minimum size=5pt]
\tikzstyle{med red}=[fill=red, draw=red, shape=circle, inner sep=0pt, minimum size=5pt]
\tikzstyle{med blue}=[fill=blue, draw=blue, shape=circle, inner sep=0pt, minimum size=5pt]
\tikzstyle{med green}=[fill={rgb,255: red,0; green,107; blue,61}, draw={rgb,255: red,0; green,107; blue,61}, shape=circle, inner sep=0pt, minimum size=5pt]
\tikzstyle{caterpillar_nodes}=[fill={rgb,255: red,255; green,200; blue,112}, draw=black, shape=circle, minimum width=1.2cm]
\tikzstyle{clique}=[fill=black, draw=black, shape=circle, minimum width=1.3cm]
\tikzstyle{expander}=[fill={rgb,255: red,175; green,183; blue,255}, draw=black, shape=circle, minimum width=4cm]
\tikzstyle{P}=[fill={rgb,255: red,255; green,123; blue,123}, draw=black, shape=circle, minimum width=3cm]
\tikzstyle{largeblack}=[fill=black, draw=black, shape=circle, minimum size=10pt]

\tikzstyle{directed}=[->, line width=1pt]
\tikzstyle{undirected}=[-, line width=1pt]
\tikzstyle{directed red}=[draw=red, ->, line width=1pt]
\tikzstyle{directed green}=[draw={rgb,255: red,0; green,128; blue,128}, ->, line width=1pt]
\tikzstyle{directed blue}=[draw=blue, ->, line width=1pt]
\tikzstyle{directed purple}=[draw={rgb,255: red,128; green,0; blue,128}, ->, line width=1pt]
\tikzstyle{undirected red}=[-, draw={rgb,255: red,222; green,45; blue,38}, line width=1pt]
\tikzstyle{undirected green}=[-, draw={rgb,255: red,0; green,107; blue,61}, line width=1pt]
\tikzstyle{undirected blue}=[-, draw=blue, line width=1pt]
\tikzstyle{undirected purple}=[-, draw={rgb,255: red,128; green,0; blue,128}, line width=1pt]
\tikzstyle{undirected dashed}=[-, line width=1pt, dashed]
\tikzstyle{orange dashed}=[-, draw={rgb,255: red,255; green,128; blue,0}, dashed, line width=1.5pt]
\tikzstyle{directed dash}=[->, dashed]
\tikzstyle{blue dashed}=[-, draw=blue, dashed, line width=1pt]
\tikzstyle{green dashed}=[-, draw={rgb,255: red,0; green,162; blue,0}, dashed, line width=1pt]
\tikzstyle{blue filled}=[-, fill={blue!20}, draw=blue, line width=1pt, opacity=0.5, tikzit fill=white]
\tikzstyle{red filled}=[-, fill={red!20}, line width=1pt, draw=red, opacity=0.5, tikzit fill=white]
\tikzstyle{green filled}=[-, line width=1pt, draw={rgb,255: red,0; green,107; blue,61}, opacity=0.5, tikzit fill=white, fill={rgb,255: red,149; green,255; blue,179}]
\tikzstyle{orange filled}=[-, fill={orange!20}, draw=orange, line width=1pt, opacity=0.5, tikzit fill=white]
\tikzstyle{undirected dashed}=[-, draw=black, dashed, line width=1.3pt]
\tikzstyle{thick}=[-, line width=3pt]
\tikzstyle{red dashed}=[-, line width=1.3pt, dashed, draw={rgb,255: red,222; green,45; blue,38}]

\usepackage{hyperref}

 \usepackage[accepted]{icml2024}

\usepackage{amsmath}
\usepackage{amssymb}
\usepackage{mathtools}
\usepackage{amsthm}

\usepackage{enumitem}

\usepackage[capitalize,noabbrev]{cleveref}

\theoremstyle{plain}
\newtheorem{theorem}{Theorem}[section]

\newtheorem{lemma}[theorem]{Lemma}
\newtheorem{corollary}[theorem]{Corollary}
\theoremstyle{definition}
\newtheorem{definition}[theorem]{Definition}

\theoremstyle{remark}

\newcommand{\vol}{\mathrm{vol}}

\newcommand{\lp}{\left (}
\newcommand{\rp}{\right )}

\newcommand{\barg}{\Bar{g}}
\newcommand{\hatg}{\widehat{g}}
\newcommand{\hatf}{\widehat{f}}
\newcommand{\norm}[1]{\left\| #1\right\|}   
\newcommand{\transpose}{\intercal}

\newcommand{\ex}[1]{\mathbb{E}[\,#1\,]}

\newcommand{\poly}{\operatorname{poly}}

\renewcommand{\deg}{\mathrm{deg}}
\newcommand{\deglist}{\mathbf{sp}^*}

\DeclareMathOperator{\spn}{span}
\DeclareMathOperator{\dmn}{dim}
\renewcommand{\leq}{\leqslant}
\renewcommand{\geq}{\geqslant}

\renewcommand{\tilde}{\widetilde}
\renewcommand{\epsilon}{\varepsilon}

\newcommand{\CGap}{C_{\ref{lem:MS22+}}}

\newcommand{\Enew}{E_\mathrm{new}}
\newcommand{\Vnew}{V_\mathrm{new}}

\newtheorem{claim}{Claim}[theorem]

\icmltitlerunning{Dynamic Spectral Clustering with Provable 
Approximation Guarantee}

\begin{document}

\twocolumn[
\icmltitle{Dynamic Spectral Clustering with Provable Approximation Guarantee}


\begin{icmlauthorlist}
\icmlauthor{Steinar Laenen}{yyy}
\icmlauthor{He Sun}{yyy}
\end{icmlauthorlist}

\icmlaffiliation{yyy}{School of Informatics, University of Edinburgh, United Kingdom}

\icmlcorrespondingauthor{Steinar Laenen}{steinar.laenen@ed.ac.uk}
\icmlcorrespondingauthor{He Sun}{h.sun@ed.ac.uk}

\icmlkeywords{graph clustering, spectral graph theory}

\vskip 0.3in
]



\printAffiliationsAndNotice{}  


\begin{abstract}
 
This paper studies clustering   algorithms for dynamically evolving graphs $\{G_t\}$, in which new edges (and potential new vertices) are added into a graph, and the underlying cluster structure of the graph can gradually change. The paper  proves that, under some mild condition on the cluster-structure, 
the clusters of the final graph $G_T$ of $n_T$ vertices at time $T$ can be well approximated by a  dynamic variant of the spectral clustering algorithm. The algorithm runs in  amortised update time $O(1)$ and  query time $o(n_T)$. 
 Experimental studies on both synthetic and real-world datasets  further  confirm the practicality of our designed algorithm.
 \end{abstract}

\section{Introduction}

For any graph $G=(V,E)$ and parameter $k\in\mathbb{N}$ as input, the objective of graph clustering is to partition the vertex set of $G$ into $k$ clusters such that vertices within each cluster are better connected than to the rest of the graph. Since large-scale graphs are commonly used to model practical datasets, designing efficient graph clustering algorithms  is an important problem in machine learning and related fields.

In practice, these large-scale graphs usually evolve over time: not only are new vertices and edges added into a graph, but the graph's clusters  could also change gradually, resulting in a new cluster-structure in the long  term.   Instead of periodically running a   clustering algorithm from scratch,  it is important to design algorithms that can quickly identify and return the new clusters in dynamically evolving graphs.

In this paper we study clustering for dynamically evolving graphs, and obtain the following results. As the first and conceptual contribution, we propose a model for dynamic graph clustering.  In contrast to the classical model for  dynamic graph algorithms~\cite{Thorup07,stoc/BeimelKMNSS22}, our proposed model considers not only edge insertions but also vertex insertions; as such the underlying graph can gradually form a new cluster-structure with a different number of clusters from the initial graph.

As the second and algorithmic result, we design a randomised  graph clustering algorithm that works in the above-mentioned model,
and our result is as follows:

\begin{theorem}[Informal statement of Theorem~\ref{thm:main_dynamic_SC}]\label{thm:informal} Let $G_1=(V_1,E_1)$ be a graph of $n_1$ vertices and $k=\widetilde{O}(1)$ clusters.\footnote{We use $\widetilde{O}(n)$ to represent $O(n \cdot\log^c
(n))$ for  constant 
$c$.}    Assume that new edges, which could be   adjacent to new vertices, are added to $G_t$ at each time $t$ to obtain $G_{t+1}$, and there are  $O(\poly(n_1))$   added  edges in total   at time $T=O(\poly(n_1))$ to form $G_T$ of  $n_T$ vertices and $k'$ clusters. Then, there is a randomised algorithm  such that    the following hold with high probability:
\begin{itemize}\itemsep 0.06cm
\item The initial $k$ clusters of $G_1$ can be approximately computed in  $\widetilde{O}(|E_1|)$ time.
\item The  new $k'$ clusters of $G_T$  can be approximately computed  with amortised  update time  $O(1)$ and query time $o(n_T)$. 
\end{itemize}
\end{theorem}

To   examine  the result, we   notice that, although the   number of clusters $k$ in   $G_1$ can be    identified with the classical  eigen-gap heuristic~\cite{nips/NgJW01,sac/Luxburg07}, computing an eigen-gap is   expensive and cannot be directly applied to  determine the   change of $k$ in  dynamically evolving graphs. Our result   shows that 
the new number of clusters $k'$  can be    computed   by a dynamic clustering algorithm with sublinear  query time. Secondly, as the running time of a clustering algorithm is at least linear in the number of edges in  $G_T$ and it takes $\Omega(n_T)$ time to output the cluster membership of all the vertices,
obtaining an $o(n_T)$ amortised query time\footnote{Throughout the paper we use $T$ to denote  query time, and $t$ as   arbitrary time throughout the sequence of graphs $\{G_t\}$.} is significant. To the best of our knowledge, our work presents the first  such result with respect to theoretical guarantees of the output clusters, and   time complexity.

Our algorithm not  only achieves strong theoretical guarantees, but also works very well in practice.   For instance, for  input graphs with 300,000 vertices and up to  490,000,000 edges generated from the stochastic block model, our algorithm runs more than 100 times faster than repeated execution of spectral clustering on the updated graphs, while obtaining a comparable clustering result.

\subsection{Overview of the Algorithm}
For any input graph $G_1$ with a well-defined cluster structure, we first construct a \emph{cluster-preserving sparsifier} $H_1$ of $G_1$, which is a sparse subgraph of $G_1$ that maintains its cluster-structure, and employ spectral clustering on $H_1$ to obtain the initial $k$ clusters of $G_1$.  After this, with a new edge arriving at  every time $t$, our designed algorithm applies two  components to track the cluster-structure of  $G_t$.

The first component is a dynamic algorithm that maintains a cluster-preserving sparsifier $H_t$ for $G_t$. Our designed algorithm is based on sampling edges with probability proportional to the degrees of their endpoints, and these edges get resampled if their degrees have significantly changed. We show that $H_t$ always 
  preserves the cluster-structure of $G_t$, and the algorithm's amortised update time  complexity is $O(1)$.

Our second component  is an algorithm that dynamically maintains a \emph{contracted graph} $\widetilde{G}_t$ of  $G_t$, and this contracted graph is used to sketch the cluster-structure of $G_t$. For the first input graph $G_1$ and the output of spectral clustering on $H_1$, our initial contracted graph $\widetilde{G}_1$ consists of $k$  super vertices with self-loops: these super  vertices correspond to the $k$ clusters in $G_1$, and  are connected by edges with weight equal to the cut values of the corresponding clusters in $H_1$. After that, when new edges (and potentially new vertices) arrive  over time, our algorithm updates $\widetilde{G}_t$ such that (new) clusters are represented by either the same super vertices, newly added vertices, or a combination of both. The algorithm further updates the edge weights between the super vertices.
With  slight increase in the number of vertices of $\widetilde{G}_t$ over time, we prove that the cluster-structure in $G_t$ is approximately preserved in $\widetilde{G}_t$. In particular, when new clusters are formed in $G_t$, this new cluster-structure of $G_t$   can be identified by the eigen-gap of $\widetilde{G}_t$'s Laplacian matrix. See Figure~\ref{fig:bucketing_main} for the illustration of our approach.

\subsection{Related work}

Our work   directly relates  to a number of works on   incremental spectral clustering algorithms~(e.g.,~\cite{DHANJAL2014440,MartinLV18,Ning}).  These works   usually   rely on analysing the change of approximate eigenvectors and don't show the approximation guarantee of the returned clusters.  Many works along this direction further employ matrix  perturbation theory in their analysis, requiring    that the total number of vertices in a graph is fixed.

Our work is also linked to   related dynamic graph algorithm problems~(e.g.,~\cite{BernsteinBGNSS022,SaranurakW19}). However, most works in   dynamic graph algorithms   focus on the design of dynamic algorithms in a \emph{general} graph, while for dynamic clustering one needs to assume the presence of cluster-structures in the initial and final graphs, 
 such that the algorithm's performance can be rigorously analysed.  Nevertheless, some of our presented techniques, like the adaptive sampling, are inspired by the dynamic graph algorithms literature.

\section{Preliminaries\label{sec:preliminaries}}

\begin{figure*}[t]
\vskip 0.2in
    \begin{center}
    \centerline{\resizebox{14.2cm}{!}{%
    \input{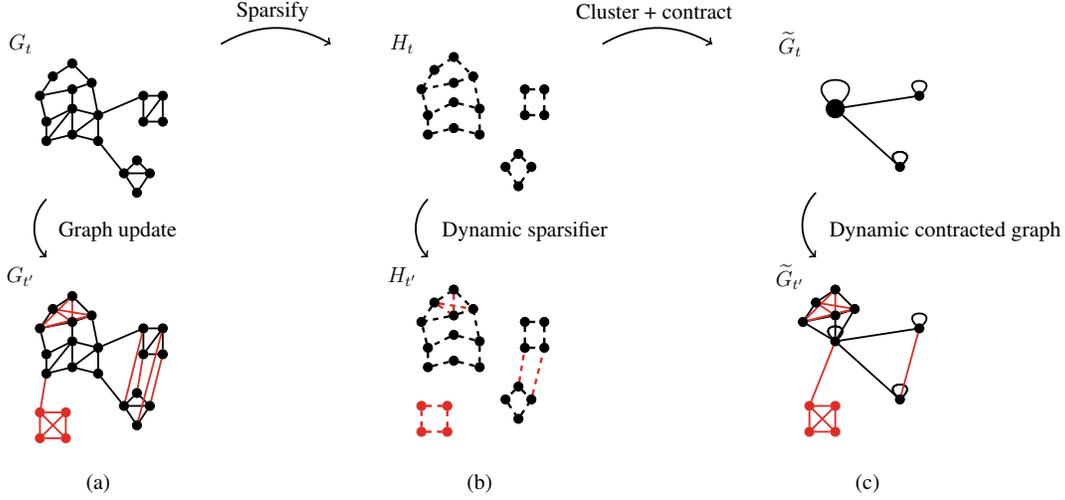}
    }}
\caption{Illustration of our technique. The black and red edges in  Figure~(a) are the   edges  in $G_t$ and the added ones in $G_{t'}$; the dashed black and red edges in  Figure~(b) are the ones added in $H_t$ and $H_{t'}$; the black and red edges in Figure~(c) are the ones in $\widetilde{G}_t$ and $\widetilde{G}_{t'}$.} \label{fig:bucketing_main}
 \end{center}
\vskip -0.2in
\end{figure*}

\subsection{Notation}

Let $G=(V,E,w)$ be an undirected graph with $|V| = n$ vertices, $|E| = m$ edges, and weight function $w: V\times V\rightarrow \mathbb{R}_{\geq 0}$. 
For any edge $e = \{u, v\} \in E$, we write $w_G(u,v)$ or $w_G(e)$ to express the   weight of $e$. For a vertex $u \in V$, we denote its  \emph{degree}  by $\deg_G(u) \triangleq \sum_{v \in V} w_G(u,v)$, and the volume for any $S \subseteq V$ is defined as $\vol_G(S)\triangleq \sum_{u\in S} \deg_{G}(u)$. For any   $S, T \subset V$, we define the \emph{cut value} between $S$ and $T$ by  $w_G(S, T) \triangleq \sum_{e \in E_G(S, T)} w_G(e)$, where $E_G(S, T)$ is the set of edges between $S$ and $T$. Moreover, for any $S\subset V$,  the conductance of $S$ is defined as 
\[
\Phi_G(S) \triangleq \frac{w_G(S, V\setminus S) }{\min\{\vol_G(S), \vol_{G}(V \setminus S)\}}
\]
if $S \neq \emptyset$, and $\Phi_G(S) = 1$ if $S = \emptyset$.  For any integer $k\geq 2$, we call 
 subsets of vertices $A_1,\ldots, A_k$ a $k$-way partition of $G$ if $\bigcup_{i=1}^k A_i = V$ and 
 $A_i\cap A_j=\emptyset$ for different $i$ and $j$. We   define the \emph{$k$-way expansion} of $G$ by
\[
    \rho_G(k) \triangleq \min_{\mathrm{partitions \:} A_1, \dots, A_k} \max_{1\leq i \leq k} \Phi_G(A_i). \]

 Our analysis is based on the spectral properties of graphs, and we list the basics of spectral graph theory. 
For a graph $G = (V, E, w)$,  let $D_G\in\mathbb{R}^{n\times n}$ be the diagonal matrix defined by $D_G(u,u) = \deg_G(u)$ for all $u \in V$. We  denote by  $A_G\in\mathbb{R}^{n\times n}$  the  \emph{adjacency matrix}  of $G$, where $A_G(u,v) = w_G(u,v)$ for all $u, v \in V$. The \emph{normalised Laplacian matrix} of $G$ is defined as $\mathcal{L}_G \triangleq I - D_G^{-1/2} A_G D_G^{-1/2}$, where $I$ is the $n \times n$ identity matrix.  The normalised Laplacian $\mathcal{L}_G$ is symmetric and real-valued, and   has $n$ real eigenvalues which we   write as  $0=\lambda_1(\mathcal{L}_G) \leq \ldots \leq \lambda_n(\mathcal{L}_G)\leq 2$; we use $f_i\in\mathbb{R}^{n}(1\leq i \leq n )$ to express the eigenvector of $\mathcal{L}_G$ corresponding to $\lambda_i$.

\begin{lemma}[Higher-order Cheeger inequality, \cite{higherCheeg}]\label{lem:Higher Cheeger}
 There is an absolute constant $C_{\ref{lem:Higher Cheeger}}$ such that it holds for any graph $G$ and $k \geq 2$ that 
\begin{equation}
    \frac{\lambda_k(\mathcal{L}_G)}{2} \leq \rho_G(k) \leq C_{\ref{lem:Higher Cheeger}} \cdot k^3 \sqrt{\lambda_k(\mathcal{L}_G)}. \label{eq:Higher Cheeger}
\end{equation}
\end{lemma}

\subsection{Spectral Clustering}  
 Spectral clustering is a popular clustering algorithm used in practice~\cite{nips/NgJW01}, and it can be described with   a few lines of code (Algorithm~\ref{alg:spectralclustering}).

\begin{algorithm}[ht]
    \caption{$\mathsf{SpectralClustering}(G, k)$}\label{alg:spectralclustering}
    
    \begin{algorithmic}[1]
    
    \STATE \textbf{Input:} {Graph $G=(V,E, w)$, number of clusters $k\in\mathbb{N}$}
    
    \STATE \textbf{Output:} Partitioning $P_1, \ldots, P_k$
    
    \STATE Compute eigenvectors $f_1, \ldots, f_k$ of $\mathcal{L}_G$
    
    \FOR{$u \in V$}
        \STATE $F(u) \leftarrow \frac{1}{\sqrt{\deg_G(u)}}\cdot (f_1(u),\ldots, f_k(u))^{\intercal}$
    \ENDFOR
    
    \STATE $P_1, \ldots, P_k \leftarrow k\mbox{-means} (\{F(u)\}_{u\in V}, k)$
    
    \STATE \textbf{Return} $P_1, \ldots, P_k$
    
    \end{algorithmic}
\end{algorithm}

 To analyse the   performance of spectral clustering, we examine the scenario in which there is a large gap between $\lambda_{k+1}(\mathcal{L}_G)$ and $\rho_G(k)$. By the higher-order Cheeger inequality,   a low value of $\rho_G(k)$ ensures that   $V$  can be partitioned into $k$ clusters, each of which has conductance at most $\rho_G(k)$; on the other hand, a large value of $\lambda_{k+1}(\mathcal{L}_G)$ implies  that  any $(k+1)$ partition of $V$ would introduce some $A\subset V$ with   $\Phi_G(A)\geq \rho_G(k+1)\geq \lambda_{k+1}(\mathcal{L}_G)/2$.  
 Based on this,  \citet{peng_partitioning_2017} introduced the parameter  
\begin{equation}\label{eq:Upsilon}
\Upsilon_G(k) \triangleq \frac{\lambda_{k+1}(\mathcal{L}_G)}{\rho_G(k)},
\end{equation}
and showed that a large value of $\Upsilon_G(k)$ is sufficient to guarantee a good performance of   spectral clustering. They further showed that, for a graph $G$ with $m$ edges, spectral clustering runs in $O(m\cdot\log^{\beta} m)$ time for constant $\beta\in\mathbb{R}^+$.

For convenience of notation, we always order the output of spectral clustering by $P_1, \ldots, P_k$ such that $\vol_{G}(P_1) \leq \ldots \leq \vol_{G}(P_k)$. 

\subsection{Model for Dynamic Graph Clustering}

We assume that the initial graph $G_1=(V_1,E_1)$ with $n_1$ vertices satisfies $\lambda_{k+1}(\mathcal{L}_{G_1})=\Omega(1)$ 
and $\rho_{G_1}(k)=O(k^{-8}\log^{-2\gamma}(n_1))$   for some constant $\gamma\in\mathbb{R}^+$. This condition is similar to lower bounding 
 $\Upsilon_{G_1}(k)$, and 
ensures that the initial input graph $G_1$ has $k$ well-defined clusters. After this, the underlying graph is updated through an edge insertion at each time, and let $G_t=(V_t, E_t)$ be the graph constructed at time   $t$.
We assume that  every edge insertion introduces at most one new vertex; as such the underlying graph is always connected, and the  number of vertices $n_t\triangleq |V_t|$ could increase over time. We further assume that,  after every $\Theta(\log^{\gamma}(n_t))$ steps, there is  time $t'$ such that 
$G_{t'} = (V_{t'}, E_{t'})$ presents a   well-defined structure of $k'$ clusters, which is characterised by 
  $\lambda_{k'+1}(\mathcal{L}_{G_{t'}}) = \Omega(1)$ and  $\rho_{G_{t'}}(k') = O(k'^{-8} \cdot \log^{-2\gamma}(n_{t'}))$ for some $k' \in \mathbb{N}$.
  
  Notice that, since both the number of vertices $n_t$ in time $t$ and the number of clusters could change, our above-defined \emph{dynamic gap assumption} allows the underlying graph to  gradually form a new cluster structure, e.g.,   $O(\log^{\gamma}(n_1))$ newly added vertices and their adjacent edges could initially form a small new cluster which gradually ``grows'' into a large one. On the other side, our assumption prevents the disappearance of the underlying graph's cluster-structure throughout the edge updates, which would make the objective function of a   clustering algorithm ill-defined.

\section{Dynamic Cluster-Preserving Sparsifiers\label{sec:dynamic_CPS}}

 A graph sparsifier is a sparse representation of an input graph that inherits certain properties of the
original dense graph, and their efficient construction  plays a key  role in designing a
number of nearly-linear time graph algorithms. However,  typical constructions of graph sparsifiers are based on  fast Laplacian solvers, making them difficult to implement in practice. To overcome this, \citet{SZ19} 
  studied a variant of graph sparsifiers for   graph clustering,  and introduced the notion of a cluster-preserving sparsifier:

\begin{definition}[Cluster-preserving sparsifier]\label{def:CPS}
     Let $G=(V, E)$ be any graph with $k$ clusters, and $\{S_i\}^{k}_{i=1}$ a $k$-way partition of $G$ corresponding to $\rho_G(k)$. We call a re-weighted subgraph $H=(V, F \subset E, w_H)$ a cluster-preserving sparsifier of $G$ if (i) $\Phi_{H}(S_i) = O(k \cdot \Phi_{G}(S_i))$ for $1 \leq i \leq k$, and (ii) $\lambda_{k+1}(\mathcal{L}_H) = \Omega(\lambda_{k+1}(\mathcal{L}_G))$.
\end{definition}

To examine the two conditions of Definition~\ref{def:CPS}, notice that 
 graph $G=(V,E)$ has exactly $k$ clusters if (i) $G$ has $k$ disjoint subsets $S_1,\ldots, S_k$ of low conductance, and (ii) any $(k + 1)$-way partition of $G$ would include some 
$A \subset V$ of high conductance,
which would be implied by a lower bound on $\lambda_{k+1}(\mathcal{L}_G)$ due to \eqref{eq:Higher Cheeger}.
With the well-known eigen-gap heuristic  and theoretical analysis on spectral clustering~\cite{peng_partitioning_2017}, these two conditions  ensure that the $k$ optimal clusters in $G$ have low conductance in $H$ as well.

\subsection{The \textsf{SZ} Algorithm}
  
We first present the algorithm in \cite{SZ19} for constructing a cluster-preserving sparsifier;  we call it the 
\textsf{SZ}  algorithm for simplicity. Given any input graph $G=(V,E)$, the algorithm computes  
\[
p_u(v)  \triangleq \min\left\{ C\cdot \frac{1}{\lambda_{k+1}(\mathcal{L}_G)}\cdot \frac{\log n}{\deg_{G}(u)},1 \right\}\]
\[  
p_v(u)  \triangleq \min\left\{ C\cdot \frac{1}{\lambda_{k+1}(\mathcal{L}_G)}\cdot \frac{\log n}{\deg_{G}(v)},1 \right\},
\]
for every   $e=\{u,v\}$, where $C\in\mathbb{R}^+$ is some constant.

Then, the algorithm samples      $e=\{u,v\}$ with probability 
$
p_e \triangleq p_u(v) + p_v(u) - p_u(v)\cdot p_v(u)$,
and sets the weight of every sampled $e=\{u,v\}$ in $H$ as 
$
w_H(u,v)\triangleq 1/p_e$.
By setting  $F$ as the set  of the sampled edges, the algorithm returns $H=(V, F, w_H)$.   \citet{SZ19} proved that, with high probability,  $H$ has $\widetilde{O}(n)$ edges and is a   cluster-preserving sparsifier of $G$.

On the other side, while Definition~\ref{def:CPS} shows that the optimal clusters $S_i~(1\leq i\leq k)$ of $G$ have low conductance in $H$, it doesn't build the connection \emph{from} the vertex sets of low conductance in $H$ \emph{to} the ones in $G$. In this paper, we prove that such a connection holds as well; this   allows us to apply spectral clustering on  $H$, and reason about the conductance of its returned clusters in $G$.

\begin{lemma}\label{lem:MS22+_GvsH}

 Let $G$ be a graph with $\Upsilon_G(k) = \Omega(k)$ for some $k\in\mathbb{N}$ with optimal clusters $\{S_i\}_{i=1}^k$, and $H$   its cluster preserving sparsifier. Let $\{P_i\}_{i=1}^k$ be the output of spectral clustering on $H$, and without loss of generality let the optimal correspondence of $P_i$ be $S_i$ for any $1\leq i \leq k$. Then, it holds with high probability for any $1\leq i\leq k$ that 
\[
  \vol_G(P_i\triangle S_i)  
  = O \lp \frac{k^2}{\Upsilon_G(k)} \rp \cdot \vol_G(S_i),
\]
\[
 \Phi_{G}(P_i) = O\lp \Phi_{G}(S_i) + \frac{k^2}{\Upsilon_G(k)} \rp,
 \]
where  $A\triangle B\triangleq (A\setminus B)\cup (B\setminus A)$.
\end{lemma}

\subsection{Construction of Dynamic Cluster-Preserving Sparsifiers}

Now we design an algorithm that constructs a cluster-preserving sparsifier under edge and vertex insertions, and   our algorithm works as follows. Initially, for the   input   $G_1$ with $n_1$ vertices,  a well-defined structure of  $k$ clusters and  
\begin{equation}\label{eq:def_tau}
\tau \geq  \frac{C}{\lambda_{k+1}(\mathcal{L}_{G_1})}
\end{equation} for some constant $C\in\mathbb{R}^+$, we run the \textsf{SZ} algorithm and obtain a cluster-preserving sparsifier of $G_1$. In addition to storing the sparsifier $H_1$ of $G_1$, the algorithm employs the vector $\deglist_1$ to store the values $\log{n_1}/\deg_{G_1}(u)$ for every vertex $u$, which are   used to sample adjacent edges of vertex $u$.  See Algorithms~\ref{alg:static_SZ} and \ref{alg:SZ_algorithm} for formal description.

\begin{algorithm}[h]
    \caption{$\mathsf{SampleEdge}(e, G, \tau)$}\label{alg:static_SZ}
    
    \begin{algorithmic}[1]
    
    \STATE \textbf{Input:} {edge $e=\{u,v\}$, graph $G=(V,E)$ of $n$ vertices, parameter $\tau\in\mathbb{R}^+$}
    
    \textbf{Output:} edge $e'$ with weight $w(e')$
    
         $p(u,v) \leftarrow p_u(v) + p_v(u) - p_u(v)\cdot p_v(u)$

        Sample $e$ with probability $p(u,v)$
    
        \IF{$e$ is sampled}
         \STATE  $e' \leftarrow e$,  $w(e') \leftarrow 1/p(u,v)$ 
        
        \ELSE
          \STATE $e' \leftarrow \emptyset$, $w(e') \leftarrow 0$ 
        \ENDIF
        
        \STATE \textbf{Return}  $e'$, $w(e')$
        \end{algorithmic}
 \end{algorithm}

\begin{algorithm}
\caption{$\mathsf{StaticSZSparsifier}(G, \tau)$}\label{alg:SZ_algorithm}
\begin{algorithmic}[1]
    \STATE \textbf{Input:} $G=(V,E)$ of $n$ vertices, parameter $\tau\in\mathbb{R}^+$
    \STATE \textbf{Output:} Cluster preserving sparsifier $H = (V, F, w_H)$,  degree list $\deglist$
    
    \STATE  $F \leftarrow \emptyset$
    
    \FOR{$e \in E$}{
       
        \STATE $e', w(e') \leftarrow \mathsf{SampleEdge}(e, G, \tau)$
    
        \STATE $F \leftarrow F \cup  e'$, $w_H(e) \leftarrow w(e')$ 
    }
    \ENDFOR
   \STATE  $\deglist \leftarrow \left\{\frac{\log{n}}{\deg_G(u)} \mid u \in V\right\}$
    
   \STATE \textbf{Return} $H$, $\deglist$
\end{algorithmic}
\end{algorithm}

\begin{algorithm}
\textbf{ \caption{$\mathsf{UpdateSparsifier}(G_t, H_t, \deglist_t, e, \tau)$}\label{alg:dynamic_update_cp_sparsifier}}
\begin{algorithmic}[1]
    
    \STATE \textbf{Input:} $G_t=(V_t,E_t)$, $H_t=(V_t, F_t, w_{H_t})$, $\deglist_t$, incoming edge $e=\{u,v\}$, parameter $\tau$
    
    \STATE \textbf{Output:}  $H_{t+1} = (V_{t+1}, F_{t+1}, w_{H_{t+1}})$, $\deglist_{t+1}$

    \STATE $V_\mathrm{new} \leftarrow \{u,v\} \setminus V_t$
    
    \STATE $G_{t+1} \leftarrow (V_t \cup V_\mathrm{new}, E_t \cup e)$
    
    \STATE $H_{t+1} \leftarrow (V_t \cup V_\mathrm{new}, F_t, w_{H_t})$

    \STATE $\deglist_{t+1} \leftarrow \deglist_t$ 

    \IF{$V_\mathrm{new} \neq \emptyset$}{
     \STATE  $e', w(e') \leftarrow \mathsf{SampleEdge}(e, G_{t+1}, \tau)$ 
    
     \STATE $F_{t+1} \leftarrow F_{t+1} \cup  e'$,  $w_{H_{t+1}}(e) \leftarrow w(e')$

                \IF{$u\in V_{\mathrm{new}}$}{ \STATE $\deglist_{t+1}(u) \leftarrow \frac{\log{n_{t+1}}}{\deg_{G_{t+1}}(u)}$}
                \ENDIF

                \IF{$v\in V_{\mathrm{new}}$}{
                \STATE $\deglist_{t+1}(u) \leftarrow \frac{\log{n_{t+1}}}{\deg_{G_{t+1}}(v)}$}
                \ENDIF
    }
    \ENDIF

     $V_\mathrm{doubled} \leftarrow \Big\{\hat{v} \in V_{t+1} \setminus V_\mathrm{new} \mid \frac{\log{n_{t+1}}}{\deg_{G_{t+1}}(\hat{v}) } > 2 \cdot \deglist_{t}(\hat{v}) \text{ or } \frac{\log{n_{t+1}}}{\deg_{G_{t+1}}(\hat{v})} < \frac{\deglist_{t}(\hat{v})}{2}\Big\}$ \label{alg:sz_dynamic:line:doubled}

    \IF{ $|V_\mathrm{doubled}|> 0$}{
    \FOR{$\hat{u} \in V_\mathrm{doubled}$}

      \STATE  $F_{t+1} \leftarrow F_{t+1} \setminus E_{H_{t+1}}(\hat{u})$

            \FOR{$\hat{e} \in E_{G_{t+1}} \mathrm{\ adjacent\ to\ } \hat{u}$}
              
                \STATE $\hat{e}', w(\hat{e}') \leftarrow \mathsf{SampleEdge}(\hat{e}, G_{t+1}, \tau)$
    
                \STATE  $F_{t+1} \leftarrow F_{t+1} \cup  \hat{e}'$, $w_{H_{t+1}}(\hat{e}) \leftarrow w(\hat{e}')$ 
            
            \ENDFOR
        \STATE $\deglist_{t+1}(\hat{u}) \leftarrow \frac{\log{n_{t+1}}}{\deg_{G_{t+1}}(\hat{u})}$ 
         
        \ENDFOR 
}
\ELSE{
    \STATE  $e', w(e') \leftarrow \mathsf{SampleEdge}(e, G_{t+1}, \tau)$
    
        \STATE       $F_{t+1} \leftarrow F_{t+1} \cup  e'$,         $w_{H_{t+1}}(e) \leftarrow w(e')$
 
}
\ENDIF

\STATE \textbf{Return} $H_{t+1}$, $\deglist_{t+1}$
\end{algorithmic}
    
\end{algorithm}

Next, given the graph $G_t$ currently constructed at time $t$, its   sparsifier $H_t$, and   edge insertion $e=\{u,v\}$, the algorithm  
compares for every vertex $w$ the parameter   $\log{n_{t+1}}/\deg_{G_{t+1}}(w)$  with $\deglist_t(w)$, the quantity used to sample the adjacent edges of $w$ the last time, and checks whether the two values change significantly. If
it is the case, then the used sampling probability is too far from the ``correct'' one when running the static \textsf{SZ}  algorithm on $G_{t+1}$, and hence we   resample all the edges adjacent to $w$ with the right sampling probability. Otherwise, we simply use the values stored in $\deglist_t$ to sample the upcoming edge $e$, and include it in $H_{t+1}$ if $e$ is sampled. See Algorithm~\ref{alg:dynamic_update_cp_sparsifier} for formal description\footnote{Notice that, since  $\lambda_{k+1}(\mathcal{L}_{G_t})=\Omega(1)$ for any graph $G_t$ exhibiting a well-defined structure of $k$ clusters and it holds for   $G_T$ at time $T=O(\poly(n_1))$   that $n_T=O(\poly(n_1))$, i.e., $\log{n_T} =O(\log{n_1})$, by setting $C$ to be a sufficiently large constant,    $\tau\cdot\log{n_1}$ is the right parameter  for defining the sampling probability at time $T=O(\poly(n_1))$.}, and Theorem~\ref{thm:dynamic_cp_sparsifier} for its performance:

\begin{theorem} \label{thm:dynamic_cp_sparsifier}
    Let $G_1=(V_1, E_1)$ be a   graph with $n_1$ vertices and a well-defined structure of $k=\widetilde{O}(1)$ clusters, and  $\{G_t\}$  the sequence of  graphs  of $\{n_t\}$ vertices constructed sequentially through an edge insertion at each time. Assuming   graph $G_T$  at time $T=O(\poly(n_1))$ has a well-defined structure of $\widetilde{O}(1)$ clusters  and $n_T=O(\mathrm{poly}(n_1))$, Algorithm~\ref{alg:dynamic_update_cp_sparsifier} returns a cluster-preserving sparsifier $H_T=(V_T, F_T, w_{H_T})$ of $G_T$ with high probability, and $|F_T| = \widetilde{O}(n_T)$. The algorithm's amortised running time is $O(1)$ per edge update.   
\end{theorem}

\section{Dynamic Spectral Clustering Algorithm\label{sec:algorithm}}

 This section presents our main dynamic spectral clustering algorithm, and is organised as follows: In Section~\ref{sec:const_cp}, we present the construction and update procedure of a contracted graph,  which is  the data structure that summarises the cluster structure of an underlying input graph and allows for quick updates to the clusters. The properties of dynamic contracted graphs are analysed in Section~\ref{sec:property_CG}. We present the main algorithm and analyse its performance in Section~\ref{sec:main_algo}.

\subsection{Construction and Update of Contracted Graphs\label{sec:const_cp}}

For any   input graph $G_t=(V_t, E_t)$ of $n_t$ vertices, its dynamic cluster-preserving sparsifier $H_t=(V_t,F_t, w_{H_t})$, and its $k$ clusters $P_1,\ldots, P_k$ returned from  running spectral clustering on $H_t$, we apply Algorithm~\ref{alg:construct_contracted_graph} to construct a contracted graph  $\widetilde{G}_t=(\widetilde{V}_t, \widetilde{E}_t, w_{\widetilde{G}_t})$ of $G_t$. Notice that we introduce the set of \emph{non-contracted vertices} $\widetilde{V}_t^\mathrm{nc} = \emptyset$, which will be used later.

\begin{algorithm}
   \caption{$\mathsf{ContractGraph}(H_t, \mathcal{P})$}\label{alg:construct_contracted_graph}
\begin{algorithmic}[1]

 \STATE \textbf{Input:} Cluster preserving sparsifier $H_t = (V_t, F_t, w_{H_t})$, partition  $\mathcal{P} = \{P_1, \ldots P_k \}$
     
    \STATE \textbf{Output:}
    Contracted graph $\tilde{G}_t=(\tilde{V}_t, \tilde{E}_t, w_{\tilde{G}_t})$

     \STATE Let $p_i$ be a representative super vertex for each cluster $P_i \in \mathcal{P}$.
    
     \STATE $\tilde{V}_t^\mathrm{c} \leftarrow \{p_i \mid P_i \in \mathcal{P}\}$, $\tilde{V}_t^\mathrm{nc} \leftarrow \emptyset$ 

     \STATE $\tilde{V}_t \leftarrow \tilde{V}_t^\mathrm{nc} \cup \tilde{V}_t^\mathrm{c}$   \label{algcontractgraph:line:end_add_vertices}

    \STATE $\tilde{E}_t \leftarrow \emptyset$

   \FOR{$\{p_i,p_j\} \in \tilde{V}_t^\mathrm{c} \times \tilde{V}_t^\mathrm{c}$}{ \label{algcontractgraph:line:start_add_edges_big_clusters}
     \STATE  $\tilde{E}_t \leftarrow \tilde{E}_t \cup \{p_i,p_j\}$
     \STATE $w_{\tilde{G}_t}(p_i,p_j) \leftarrow w_{H_t}(P_i, P_j)$}
    \ENDFOR 

    \label{algcontractgraph:line:end_add_edges_big_clusters}

    \STATE \textbf{Return} $\tilde{G}_t = (\tilde{V}_t, \tilde{E}_t, w_{\tilde{G}_t})$

\end{algorithmic}
\end{algorithm}

\begin{lemma}\label{lem:time_complexity_contract_graph}
     The algorithm
    $\mathsf{ContractGraph}(H_t, \mathcal{P})$ returns   $\tilde{G}_t = (\tilde{V}_t, \tilde{E}_t, w_{\tilde{G}_t})$ in $O\lp |F_t| \rp$ time.
\end{lemma}

Next we discuss how the contracted graph is updated under   edge and vertex insertions.
Given the graph $G_t=(V_t, E_t)$  with $n_t$ vertices that satisfies $\lambda_{k+1}(\mathcal{L}_{G_t})=\Omega(1)$ 
and $\rho_{G_t}(k)=O(k^{-8}\log^{-2\gamma}(n_t))$   for some constant $\gamma\in\mathbb{R}^+$, its cluster-preserving sparsifier $H_t=(V_t, F_t, w_{H_t})$,   the corresponding contracted graph   $\widetilde{G}_t=(\widetilde{V}_t,\widetilde{E}_t, w_{\widetilde{G}_t})$, and the upcoming edge insertion $e=\{u,v\}$, we construct $\widetilde{G}_{t+1}$ from  $\widetilde{G}_t$  as follows:
\begin{itemize}\itemsep 0.06cm
\item If either $u$ or $v$ is a new vertex, the algorithm adds the vertex to $\widetilde{G}_t$ as a non-contracted vertex. The algorithm sets $V_{\mathrm{new}}= \{u,v\}\setminus V_t$, and $V_{t+1} = V_t\cup V_{\mathrm{new}}$.
\item  For every existing vertex $w\in\{u,v\}\setminus V_{\mathrm{new}}$ that belongs to some $P_i$,  the algorithm checks whether $\deg_{G_{t+1}}(w)>2\cdot \deg_{G_r}(w)$, where $\deg_{G_r}(w)$ for $r \leq t$ is the degree of $w$ when the contracted graph was constructed. If it is the case, the algorithm pulls  $w$ out of $p_i$, and adds it to $\widetilde{V}_{t+1}$, i.e., the uses a  single vertex in $\widetilde{G}_{t+1}$ to represent   $w$.    
\item The algorithm adjusts the edge weights in the contracted graph based on the type of the vertices. For  instance, the algorithm sets $w_{\widetilde{G}_{t+1}}(u,v)=1$
if both of $u$ and $v$ are non-contracted vertices,  and   decreases the value of $w_{\widetilde{G}_{t+1}}(P_u,P_u)$ if vertex $u$ pulls out of   $P_u\in\mathcal{P}$.
\end{itemize}
See Algorithm~\ref{alg:update_contracted_graph} in the appendix for the formal description of the algorithm $\mathsf{UpdateContractedGraph}(G_t, \widetilde{G}_t, e)$.

\begin{lemma}\label{lem:running_time_dynamic_update_contracted_graph}
    The amortised time complexity of $\mathsf{UpdateContractedGraph}(G_t, \widetilde{G}_t, e)$ is $O(1)$.
\end{lemma}

\subsection{Properties of the Contracted Graph \label{sec:property_CG}}
Now we analyse the properties of the contracted graph.   Since the amortised time complexity for every edge update~(Theorem~\ref{thm:dynamic_cp_sparsifier} and Lemma~\ref{lem:running_time_dynamic_update_contracted_graph}) remains valid when we consider a sequence of edge updates at every time, without loss of generality let $G_{t'} = (V_t \cup \Vnew, E_t \cup \Enew)$ be the graph after a sequence of edge updates from $G_t=(V_t, E_t)$ with $n_t$ vertices, and $\tilde{G}_{t'}$ be the contracted graph of $G_{t'}$ constructed by sequentially running $\mathsf{UpdateContractedGraph}$ for each   $e \in \Enew$.   
We assume that $|\Enew| \leq \log^{\gamma}(n_t)$ for some $\gamma\in\mathbb{R}^+$.

 We first prove that the clusters returned by spectral clustering on $H_t$ also have low conductance in $G_t$. Notice that, as the underlying graph $G_t$ could be dense over time, running a clustering algorithm on its sparsifier $H_t$ with $\widetilde{O}(n_t)$ edges is crucial   to achieve the algorithm's quick update time.

\begin{lemma}\label{lem:conductanceHvsG}
     It holds with high probability that  
    $
    \Phi_{H_t}(P_i) = O \lp k^2 \cdot \rho_{G_t}(k) \rp$
    and
      $
        \Phi_{G_t}(P_i) =  O \lp k^2 \cdot \rho_{G_t}(k) \rp$ for all $P_i \in \mathcal{P}$.
\end{lemma}
 Next, we define the event $\mathcal{E}_1$ that
\[
\Phi_{H_t}(P_i) = O \lp k^{-6} \cdot \log^{-2\gamma}(n_t) \rp
\]
and 
\[
\Phi_{G_t}(P_i) =  O \lp k^{-6} \cdot \log^{-2\gamma}(n_t) \rp 
\]
hold for all $P_i\in \mathcal{P}$.
By  the fact that  $\lambda_{k+1}(G_t) = \Omega(1)$, $\rho_{G_t}(k) = O(k^{-8} \cdot \log^{-2\gamma}(n_t))$ and Lemma~\ref{lem:conductanceHvsG},   $\mathcal{E}_1$ holds with high probability.  We further define the event $\mathcal{E}_2$ that
\[
(1/2)\cdot \deg_{G_t} (u) \leq \deg_{H_t}(u) \leq (3/2)\cdot \deg_{G_t}(u) 
\]
hold for all $u\in V_t$,
and know from the proof of Theorem~\ref{thm:dynamic_cp_sparsifier} that $\mathcal{E}_2$ holds with high probability. In the following   we assume  that both of $\mathcal{E}_1$ and  $\mathcal{E}_2$ happen.

Next, we study  the relationship between the cluster-structure in $G_{t'}$ and the one in $\widetilde{G}_{t'}$. 
Recall that the number of vertices in  $\widetilde{G}_{t'}$ is   much smaller than the one in $G_{t'}$. We first prove that there are $\ell$ disjoint vertex sets of low conductance in $G_{t'}$ if and only if there are $\ell$ such vertex sets  in $\widetilde{G}_{t'}$.

\begin{lemma}\label{lem:conductance_full>contract}  The following statements hold:   
\begin{itemize}\itemsep 0.06cm
\item 
If $\rho_{G_{t'}}(\ell) \leq \log^{-\alpha}(n_{t'})$ holds for some $\ell \in \mathbb{N}$  and $\alpha > 0$, then $\rho_{\tilde{G}_{t'}}(\ell) = \max \left\{O \lp  \log^{-0.9\alpha}(n_{t'}) \rp, O\lp k^{-6} \cdot \log^{-\gamma}(n_{t'}) \rp\right\}$.
\item If $\rho_{\tilde{G}_{t'}}(\ell) \leq \log^{-\delta}(n_{t'})$  holds for some $\ell \in \mathbb{N}$ and  $\delta > 0$ , then  $\rho_{G_{t'}}(\ell) = \max \left\{O \lp \log^{-\delta}(n_{t'}) \rp, O\lp k^{-6} \cdot \log^{-\gamma}(n_{t'})\rp \right\}$. 
\end{itemize}
\end{lemma}
 
Secondly, we show that there is a close connection between $\lambda_{\ell+1}(\mathcal{L}_{G_{t'}})$ and $\lambda_{\ell+1}(\mathcal{L}_{\widetilde{G}_{t'}})$ for any $\ell\in\mathbb{N}$.

\begin{lemma}\label{lem:lambda_k+1 from tilde(G) to G} The following statements hold:
\vspace{-0.1cm}
\begin{itemize}\itemsep 0.06cm
\item 
   If $\lambda_{\ell+1}(\mathcal{L}_{\tilde{G}_{t'}}) = \Omega(1)$ for some $\ell \in \mathbb{N}$, then  $\lambda_{\ell+1}\lp \mathcal{L}_{G_{t'}}\rp = \Omega \lp \log^{-\alpha}(n_{t'})/{\ell^6} \rp$ for constant $\alpha > 0$.
\item  If   
      $\lambda_{\ell+1}\lp \mathcal{L}_{G_{t'}} \rp= \Omega(1)$ holds for some $\ell\in\mathbb{N}$, then $\lambda_{\ell+1}(\mathcal{L}_{\tilde{G}_{t'}}) = \Omega(1)$.
\end{itemize}
\end{lemma}

  Lemmas~\ref{lem:conductance_full>contract} and \ref{lem:lambda_k+1 from tilde(G) to G} imply that the cluster-structures in $G_{t'}$ and $\widetilde{G}_{t'}$ are approximately preserved.

\subsection{Main Algorithm \label{sec:main_algo}}
Our main algorithm consists of the preprocessing stage, update stage, and query stage. They are   described as follows:
 
 \emph{Preprocessing Stage.} For the initial input graph $G_1=(V_1, E_1)$, we
apply (i)  $\mathsf{StaticSZSparsifier}(G_1, \tau)$ to obtain $H_1 = (V_1, F_1, w_{H_1})$,  (ii) $\mathsf{SpectralClustering}(H_1, k)$ to obtain initial partition $\mathcal{P} = \{P_1, \ldots P_k\}$, and (iii) $\mathsf{ContractGraph}(H_1, \mathcal{P})$ to obtain $\widetilde{G}_1 = (\widetilde{V}_1, \tilde{E}_1)$.

\emph{Update Stage.} 
 When a new edge arrives at time $t$, we apply Algorithm~\ref{alg:dynamic_update_cp_sparsifier} and the update procedure of the contracted graph (Section~\ref{sec:const_cp})    to dynamically maintain $H_t$ and $\widetilde{G}_t$.

\emph{Query Stage.}  When a query for a new clustering   starts at time $T$, the algorithm performs the following operations, where $\gamma$ is the constant satisfying $\gamma>\beta$ and $\gamma>0.9\alpha$: 
 \begin{itemize}\itemsep 0.06cm
\item For $r$ being the last time   at which $\widetilde{G}_t$ is recomputed, the algorithm checks if $T-r \leq \log^{\gamma}(n_r)$, i.e., the number of added edges after the last reconstruction of the contracted graph is less than $\log^{\gamma}(n_r)$. If it is the case,  then the algorithm  runs spectral clustering on the contracted graph $\widetilde{G}_T$.
\item Otherwise, the algorithm runs  spectral clustering on $H_T$. It also recomputes $\widetilde{G}_T$, by first computing $\widetilde{G}_{r'}$, where $r'$ is the last time at which the dynamic gap assumption holds, and updating $\widetilde{G}_{r'}$ to $\widetilde{G}_T$ with the edge updates between time $r'$ and $T$.
\end{itemize}
See Algorithm~\ref{alg:sc_query} for  formal description.

\begin{algorithm}
   \caption{$\mathsf{QuerySpecClustering}(G_T, H_T, 
   \widetilde{G}_T, \gamma, \ell)$}\label{alg:sc_query}
\begin{algorithmic}[1]

    \STATE \textbf{Input:} Graphs $G_T$,   $H_T$,  and   $\widetilde{G}_T$,     $\gamma \in \mathbb{R}^{+}$, and  $\ell \in \mathbb{N}$
     
    \STATE \textbf{Output:}
    Partition $\mathcal{P} = \{P_1, \ldots P_\ell \}$

    \STATE  Let $r$ be the last time   at which $\widetilde{G}_T$ is recomputed.

    \IF{ $T - r \leq \log^{\gamma}(n_r)$} \label{alg:line:querysc:start_cluster_contract} 
    {
    \STATE  $P_1, \ldots, P_\ell \leftarrow \mathsf{SpectralClustering}(\widetilde{G}_{T}, \ell)$
    
    \STATE \textbf{Return} $\{P_1, \ldots, P_\ell\}$ \label{alg:line:querysc:end_cluster_contract}
    }

    \ELSE\label{alg:line:querysc:start_recompute_contract}
    {

    \STATE $P_1, \ldots, P_\ell \leftarrow \mathsf{SpectralClustering}(H_{T}, \ell)$ \label{alg:line:querysc:cluster_H}

    \STATE  Recompute $\widetilde{G}_{r'}$, where $r'$ is the last time   at which the dynamic gap assumption holds
    
    \STATE  Update $\widetilde{G}_{r'}$ to $\widetilde{G}_T$ with the edge updates between time $r'$ and $T$

     \STATE \textbf{Return} $\{P_1, \ldots, P_\ell\}$
    
    } 
    \ENDIF 
    \label{alg:line:querysc:end_recompute_contract}
    
\end{algorithmic}
\end{algorithm}

\begin{theorem}\label{thm:main_dynamic_SC}  Let $G_1=(V_1, E_1)$ be a graph with $n_1$ vertices and $k=\widetilde{O}(1)$ clusters, and  $\{G_t\}$  the sequence of  graphs  of $\{n_t\}$ vertices constructed     through   an edge insertion at each time   satisfying the dynamic gap assumption.
Assume that $G_T$ at query time $T$ has   $\ell$ clusters, i.e.,  
$\lambda_{\ell+1}(\mathcal{L}_{G_{T}}) = \Omega(1)$ and $\rho_{G_{T}}(\ell) = O(\ell^{-1} \log^{-\alpha}(n_{T}))$ for $\alpha \in \mathbb{R}^+$. Then,   with high probability
Algorithm~\ref{alg:sc_query}
    returns   $P_1, \ldots P_\ell$ with $
 \Phi_{G_T}(P_i) = O\lp \ell \cdot \log^{-0.9\alpha}(n_{T}) \rp
 $     for every $1\leq i\leq \ell$.
 The algorithm's running time for returning the  clusters of $G_1$ is $\widetilde{O}(|E_1|)$. Afterwards, the algorithm's  amortised update time   is $O(1)$, and   amortised query time is $o(n_T)$. 
\end{theorem}

\begin{proof}
The algorithm's running time and approximation guarantee on $G_1$ follows from \cite{MS22}, so we only need to analyse the dynamic update stage.
 We first analyse the conductance of every output $P_i$. Notice that, if Lines~\ref{alg:line:querysc:start_cluster_contract}--\ref{alg:line:querysc:end_cluster_contract} of Algorithm~\ref{alg:sc_query} are executed, then   by Lemmas~\ref{lem:MS22+_GvsH}, \ref{lem:conductance_full>contract} and \ref{lem:lambda_k+1 from tilde(G) to G} the approximation guarantee holds. Otherwise,  Lines~\ref{alg:line:querysc:start_recompute_contract}--\ref{alg:line:querysc:end_recompute_contract}  are executed, then by the dynamic gap condition and Lemma~\ref{lem:MS22+_GvsH} the approximation guarantee holds as well.

  Next, we prove the running time guarantee. The $O(1)$ amortised update time of $H_t$ and $\widetilde{G}_t$ follows by Theorem~\ref{thm:dynamic_cp_sparsifier} and Lemma~\ref{lem:running_time_dynamic_update_contracted_graph}. For the query at time $T$, notice that if Lines~\ref{alg:line:querysc:start_cluster_contract}--\ref{alg:line:querysc:end_cluster_contract} are executed, then the query time is at most $O(|\widetilde{V}_T|^3) = O((k + \log^{\gamma}(n_T))^3) = \widetilde{O}(1)$. Note, the super vertices are used as sketches to quickly update the cluster assignment of each vertex;  otherwise,    Lines~\ref{alg:line:querysc:start_recompute_contract}--\ref{alg:line:querysc:end_recompute_contract}   are executed, and the query time is dominated by spectral clustering's time complexity of   $O \lp n_T \cdot \log^{\beta}(n_T) \rp$. Since this only happens every $\log^{\gamma}(n_r) = O(\log^{\gamma}(n_T))$ edge updates, the amortised query time is $O(n_T \cdot \log^{\beta - \gamma}(n_T)) = o(n_T)$.
 
 Finally, we show that the number of clusters $\ell$ can be identified with our claimed time complexity. Notice that, if Lines~\ref{alg:line:querysc:start_cluster_contract}--\ref{alg:line:querysc:end_cluster_contract} of the algorithm are executed, then by Lemmas~\ref{lem:conductance_full>contract} and~\ref{lem:lambda_k+1 from tilde(G) to G} we can detect the spectral gap in $G_T$ using $\widetilde{G}_T$; hence we can choose $\ell$ in $o(n_t)$ time. Otherwise,    Lines~\ref{alg:line:querysc:start_recompute_contract}--\ref{alg:line:querysc:end_recompute_contract}   are executed. In this case, we   run spectral clustering with different  values of $\ell'$ and find the correct value of $\ell$ (The same procedure is done to recompute $\widetilde{G}_{r'}$). Since there are $\widetilde{O}(1)$ clusters in total,  we achieve the same query time guarantee. 
\end{proof}

\begin{figure*}[t]
\vskip 0in
\begin{center}
\centerline{\includegraphics[width=0.65\linewidth]{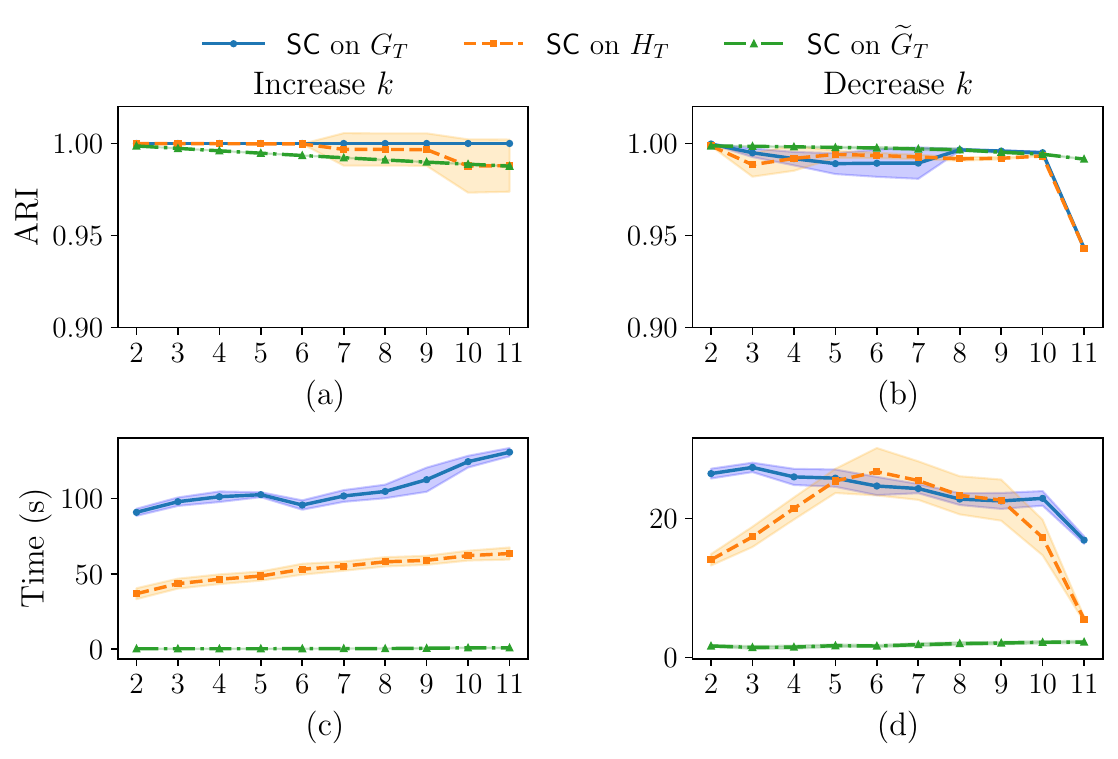}}
\caption{Results on the two versions of our dynamic SBM.  Figures (a) and (b) report the average ARI score  at each time $T$ for the clustering results on $G_T$, $H_T$, and $\widetilde{G}_T$;  Figures (c) and (d) report the running time in seconds  at each time $T$. Shaded regions indicate the standard deviation. 
\label{fig:results_sbm}}
\end{center}
\vskip -0.2in
\end{figure*}

\section{Experiments\label{sec:experiments} }

We experimentally evaluate the performance of our algorithm on synthetic and real-world datasets. We report the clustering accuracy of all tested algorithms using the Adjusted Rand Index~(ARI)~\cite{ari}, and compute the average and standard deviation over 10 independent runs.  Algorithms were implemented in Python 3.12.1 and experiments were performed using a Lenovo ThinkPad T15G, with an Intel(R) Xeon(R) W-10855M CPU@2.80GHz processor and 126 GB RAM.  Our code can be downloaded from \href{https://github.com/SteinarLaenen/Dynamic-Spectral-Clustering-With-Provable-Approximation-Guarantee}{https://github.com/steinarlaenen/Dynamic-Spectral-Clustering-With-Provable-Approximation-Guarantee}.

\subsection{Results on Synthetic Data}

We study graphs generated from the stochastic block model~(SBM), and introduce two dynamic extensions to generate new clusters and merge existing clusters.

\emph{SBM with increasing number of clusters.} We generate   the first graph $G_1$   based on the standard SBM, and  set $k=10$  and   the number of vertices in each cluster $\{S_i\}_{i=1}^k$ as $n_k = 10,000$. For every pair of $u \in S_i$ and $v \in S_j$ we include edge $\{u,v\}$ with probability $p$ if $i=j$, and with  probability $q$ if $i \neq j$.

 To update the graph, we generate a batch of edge updates in two steps: first, we randomly select a subset $Q \subset V(G_1)$ such that $|Q| = n_\mathrm{new} = 400$, and for any $u, v \in Q$ we include edge $e = \{u,v\}$ in the graph with probability $r_1$; setting $r_1$ sufficiently large ensures that the set $Q$ forms a new cluster in the graph. Second, for any $u,v \in V(G_1)$ we include edge $e = \{u,v\}$ with probability $s$. The edges sampled from   these two processes form one edge update batch. We sample 10 such batches (ensuring no new clusters overlap), each inducing a new cluster and additional noise.  

 To cluster each   $G_T$, we run spectral clustering (\textsf{SC}) on three   graphs: 
 \begin{enumerate}
     \item   We run spectral clustering on the full graph $G_T$.
     \item We construct the contracted graph $\widetilde{G}_1$ at time $T=1$, and incrementally update $\tilde{G}_1$ using the procedure described in Section~\ref{sec:const_cp}. Then, we run spectral clustering on each $\widetilde{G}_T$. 
     \item We construct a cluster-preserving sparsifier $H_1$ using Algorithm~\ref{alg:SZ_algorithm}, which we dynamically update using Algorithm~\ref{alg:dynamic_update_cp_sparsifier} with sampling parameter $\tau=3$, and cluster each subsequent $H_T$.  
\end{enumerate}
At each time $T$, we run spectral clustering with $k=10+T-1$ on all three graphs, and report the running times and ARI scores.
We set $p=0.1$, $q=0.01$, $r_1=0.95$, and $s=0.00001$, and plot the results in the left plots of Figure~\ref{fig:results_sbm}.
 We can see that at every time $T$,   spectral clustering on $G_T$ returns the perfect clustering, and   spectral clustering on $\widetilde{G}_T$ and $H_T$ returns marginally worse clustering results. On the running time, we see that running spectral clustering on $G_T$, $H_T$ and $\widetilde{G}_T$ takes around $100$ seconds, $50$ seconds, and less than $1$ second respectively.
  This highlights that our algorithm returns nearly-optimal clusters with much faster running time than running spectral clustering on $G_T$ or $H_T$.

\begin{figure*}[t]
\vskip 0in
\begin{center}
\centerline{\includegraphics[width=0.65\linewidth]{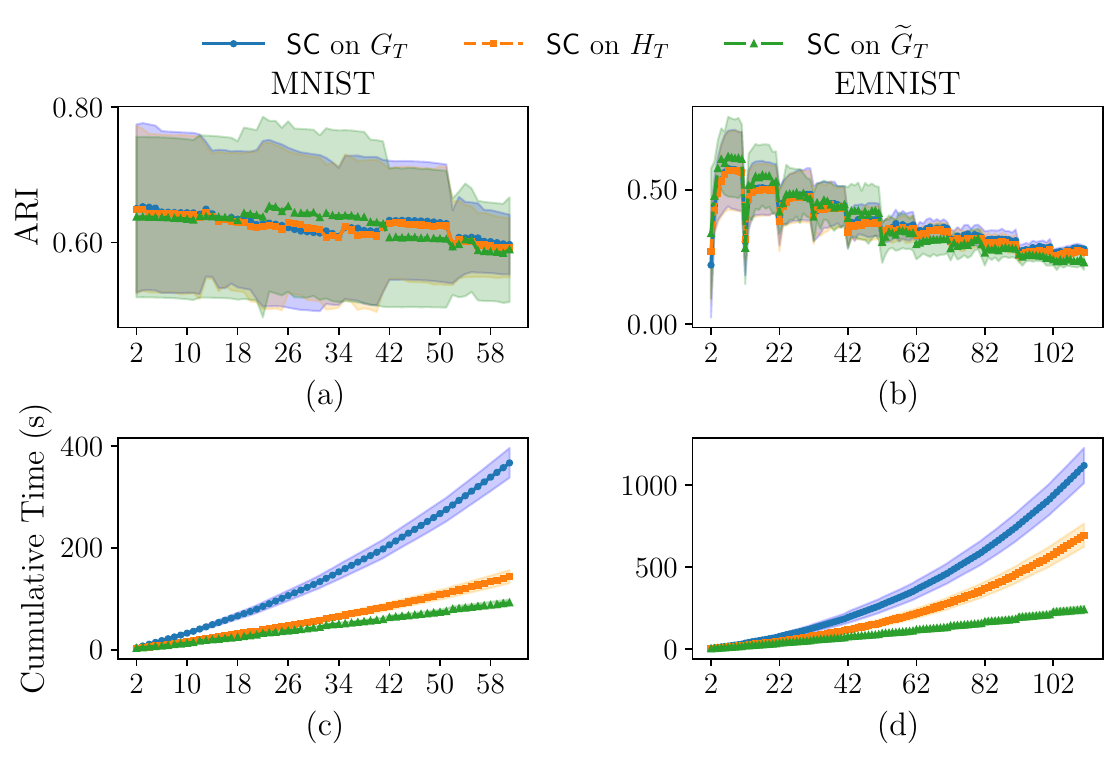}}
\caption{Results on MNIST and EMNIST.  Figures (a) and (b) report the average ARI scores  at each time $T$ for the clustering results on $G_T$, $H_T$, and $\widetilde{G}_T$; Figures (c) and (d) report the average cumulative running time in seconds at each time $T$. Shaded regions indicate the standard deviation.}\label{fig:results_mnist_emnist}
\end{center}
\vskip -0.25in
\end{figure*}

 Next, we compare the spectral gaps of $\mathcal{L}_{G_T}$ and $\mathcal{L}_{\widetilde{G}_T}$ for every $T$,   and Table~\ref{tab:spectral_gap_contract_vs_full_increase_k} reports that these   gaps are well preserved.   This demonstrates that, as what we prove earlier, the new cluster-structure of $G_T$ can be indeed identified from $\widetilde{G}_T$.

\begin{table}[h]
\caption{Spectral gaps in $\mathcal{L}_{G_T}$ and $\mathcal{L}_{\widetilde{G}_T}$ for SBM with increasing number of clusters. We report     $\lambda_{k+T}(\mathcal{L}_{G_T}) / \lambda_{k+T -1}(\mathcal{L}_{G_T})$ and $\lambda_{k+T}(\mathcal{L}_{\widetilde{G}_T}) / \lambda_{k+T -1}(\mathcal{L}_{\widetilde{G}_T})$ at each time $T$.}
\centering
\vskip 0.1in
\resizebox{\columnwidth}{!}{%
\begin{tabular}{@{}lllllllllll@{}}
\toprule
\textbf{$T$}    & 2 & 3     & 4     & 5     & 6     & 7     & 8     & 9     & 10     & 11    \\ \midrule
$G_T$             & $6.3$                                              & $5.8$ & $5.8$ & $5.7$ & $5.7$ & $5.7$ & $5.6$ & $5.6$ & $5.6$ & $5.5$ \\
$\widetilde{G}_T$ & $9.3$                                              & $9.2$ & $9.0$ & $8.7$ & $8.5$ & $8.1$ & $7.8$ & $7.5$ & $7.1$ & $6.8$   \\ \bottomrule
\end{tabular}%
}
\label{tab:spectral_gap_contract_vs_full_increase_k} 
\end{table}

 \emph{SBM with decreasing number of clusters.} We set $k=25$, and the first graph $G_1$ is generated based on the standard SBM with parameters $p$ and $q$. 
 For clusters $\{S_i\}_{i=1}^5$ we set $|S_i| = 20,000$, and for $\{S_i\}_{i=6}^{25}$ we set $|S_i| = 500$; hence  there are $5$ large and $20$ small clusters.

 To update the graph, we generate a batch of edge updates as follows: we randomly choose two clusters $S_i$ and $S_j$ such that $|S_i| = |S_j| = 500$, and for any $u \in S_i$ and $v \in S_j$ we include edge $e=\{u,v\}$ in the graph with probability $r_2$. Setting $r_2$ sufficiently large ensures that clusters $S_i$ and $S_j$ merge. Similarly as before, for any $u,v \in V(G_1)$ we also include edge $e = \{u,v\}$ with probability $s$. All the edges sampled by these two processes form a single batch update. We sample $10$ such batches, and there are $k=15$ clusters at final time  $T=11$. At each time $T$, we run spectral clustering with $k=25-T+1$ on all three graphs. We set $p=0.1$, $q=0.001$, $r_2=0.95$, and $s=0.00001$, and plot the results in the right plots of Figure~\ref{fig:results_sbm}.

Similar to the SBM with increasing number of clusters, at every time $T$,   spectral clustering   on all three graphs returns similar results.  We further see that spectral clustering on $\widetilde{G}_T$ has lower running time than the one  on $G_T$ and $H_T$. 
The spectral gaps in $G_T$ and $\widetilde{G}_T$ are reported  in Table~\ref{tab:spectral_gap_contract_vs_full_decrease_k}.

\begin{table}[h]
\caption{Spectral gaps in $\mathcal{L}_{G_T}$ and $\mathcal{L}_{\widetilde{G}_T}$ for SBM with decreasing number of clusters.  We report   $\lambda_{k-T+2}(\mathcal{L}_{G_T}) / \lambda_{k-T+1}(\mathcal{L}_{G_T})$ and $\lambda_{k-T+2}(\mathcal{L}_{\widetilde{G}_T}) / \lambda_{k-T+1}(\mathcal{L}_{\widetilde{G}_T})$ at each time $T$.}
\centering
\vskip 0.1in
\resizebox{\columnwidth}{!}{%
\begin{tabular}{@{}lllllllllll@{}}
\toprule
\textbf{$T$}    & 2 & 3     & 4     & 5     & 6     & 7     & 8     & 9     & 10     & 11    \\ \midrule
$G_T$                                  & $4.5$      & $4.4$ & $4.2$ & $4.0$ & $4.0$ & $3.9$ & $3.8$ & $3.8$ & $3.8$ & $3.6$ \\
$\widetilde{G}_T$                      & $8.3$      & $8.0$ & $7.5$ & $7.4$ & $7.4$ & $7.0$ & $6.8$ & $6.3$ & $5.8$ & $5.4$ \\ \bottomrule
\end{tabular}%
}
\vskip -0.05in
\label{tab:spectral_gap_contract_vs_full_decrease_k}
\end{table}

\subsection{Results on Real-World Data}

 We further evaluate our algorithm on 
the MNIST dataset~\cite{lecun_mnist_1998}, which consists of  10 classes of handwritten digits and has $70,000$ images,  and the ``letter'' subset of the EMNIST dataset~\cite{cohen_afshar_tapson_schaik_2017},  which consists of  26 classes of handwritten letters and has $145,600$ images.  We construct a $k$-nearest neighbour graph for each dataset, and set $k=100$~(resp. $k=200$) for   MNIST~(resp.  EMNIST).

 We   select four classes~(clusters) at random; the chosen vertices and  adjacent edges  in the $k$-nearest neighbour  graph form    $G_1$. To construct the sequence of updates, we select one new cluster~(resp. two) at random for   MNIST~(resp. EMNIST), and add the edges inside the new cluster as well as   the ones between the new and existing clusters. We randomly partition these new edges into 10 batches of equal size,  and add these to the graph sequentially. We recompute  $\tilde{G}_T$ after one  class~(resp. two) is streamed for  
MNIST~(resp. EMNIST), and report the results in  Figure~\ref{fig:results_mnist_emnist}.  The update/reconstruction time is  included in the running time.

Our experiments on real-world data  further confirm   that, as the size of the  underlying graph  and its   number of clusters increase over time, our designed algorithm has much lower running time  compared with   repeated execution  of spectral clustering, while producing   comparable clustering results.

\section*{Impact Statement}

This paper presents work whose goal is to advance the field of Machine Learning. There are many potential societal consequences of our work, none  of which we feel must be specifically highlighted here.

\section*{Acknowledgements}

This work is supported by  an EPSRC Early Career Fellowship~(EP/T00729X/1). Part of this work was done when He Sun was visiting the Simons
Institute for the Theory of Computing in Fall 2023.

\bibliography{reference}

\begin{thebibliography}{18}
\providecommand{\natexlab}[1]{#1}
\providecommand{\url}[1]{\texttt{#1}}
\expandafter\ifx\csname urlstyle\endcsname\relax
  \providecommand{\doi}[1]{doi: #1}\else
  \providecommand{\doi}{doi: \begingroup \urlstyle{rm}\Url}\fi

\bibitem[Beimel et~al.(2022)Beimel, Kaplan, Mansour, Nissim, Saranurak, and Stemmer]{stoc/BeimelKMNSS22}
Beimel, A., Kaplan, H., Mansour, Y., Nissim, K., Saranurak, T., and Stemmer, U.
\newblock Dynamic algorithms against an adaptive adversary: generic constructions and lower bounds.
\newblock In \emph{\STOC{54th}{22}}, pp.\  1671--1684, 2022.

\bibitem[Bernstein et~al.(2022)Bernstein, van~den Brand, Gutenberg, Nanongkai, Saranurak, Sidford, and Sun]{BernsteinBGNSS022}
Bernstein, A., van~den Brand, J., Gutenberg, M.~P., Nanongkai, D., Saranurak, T., Sidford, A., and Sun, H.
\newblock Fully-dynamic graph sparsifiers against an adaptive adversary.
\newblock In \emph{\ICALP{49th}{22}}, pp.\  20:1--20:20, 2022.

\bibitem[Chung \& Lu(2006)Chung and Lu]{chung2006concentration}
Chung, F. and Lu, L.
\newblock Concentration inequalities and martingale inequalities: a survey.
\newblock \emph{Internet mathematics}, 3\penalty0 (1):\penalty0 79--127, 2006.

\bibitem[Cohen et~al.(2017)Cohen, Afshar, Tapson, and Schaik]{cohen_afshar_tapson_schaik_2017}
Cohen, G., Afshar, S., Tapson, J., and Schaik, A.~V.
\newblock Emnist: Extending {MNIST} to handwritten letters.
\newblock In \emph{2017 International Joint Conference on Neural Networks (IJCNN)}, pp.\  2921--2926, 2017.

\bibitem[Dhanjal et~al.(2014)Dhanjal, Gaudel, and Clémençon]{DHANJAL2014440}
Dhanjal, C., Gaudel, R., and Clémençon, S.
\newblock Efficient eigen-updating for spectral graph clustering.
\newblock \emph{Neurocomputing}, 131:\penalty0 440--452, 2014.

\bibitem[Lecun et~al.(1998)Lecun, Bottou, Bengio, and Haffner]{lecun_mnist_1998}
Lecun, Y., Bottou, L., Bengio, Y., and Haffner, P.
\newblock Gradient-based learning applied to document recognition.
\newblock \emph{Proceedings of the IEEE}, 86\penalty0 (11):\penalty0 2278--2324, 1998.

\bibitem[Lee et~al.(2014)Lee, Gharan, and Trevisan]{higherCheeg}
Lee, J.~R., Gharan, S.~O., and Trevisan, L.
\newblock Multiway spectral partitioning and higher-order {C}heeger inequalities.
\newblock \emph{Journal of the ACM}, 61\penalty0 (6):\penalty0 1--30, 2014.

\bibitem[Macgregor \& Sun(2022)Macgregor and Sun]{MS22}
Macgregor, P. and Sun, H.
\newblock A tighter analysis of spectral clustering, and beyond.
\newblock In \emph{\ICML{39th}{22}}, pp.\  14717--14742, 2022.

\bibitem[Martin et~al.(2018)Martin, Loukas, and Vandergheynst]{MartinLV18}
Martin, L., Loukas, A., and Vandergheynst, P.
\newblock Fast approximate spectral clustering for dynamic networks.
\newblock In \emph{\ICML{35th}{18}}, pp.\  3420--3429, 2018.

\bibitem[Ng et~al.(2001)Ng, Jordan, and Weiss]{nips/NgJW01}
Ng, A.~Y., Jordan, M.~I., and Weiss, Y.
\newblock On spectral clustering: Analysis and an algorithm.
\newblock In \emph{\NIPS{15}{01}}, pp.\  849--856, 2001.

\bibitem[Ning et~al.(2007)Ning, Xu, Chi, Gong, and Huang]{Ning}
Ning, H., Xu, W., Chi, Y., Gong, Y., and Huang, T.
\newblock Incremental spectral clustering with application to monitoring of evolving blog communities.
\newblock In \emph{Proceedings of the 2007 SIAM International Conference on Data Mining (SDM)}, pp.\  261--272, 2007.

\bibitem[Peng et~al.(2017)Peng, Sun, and Zanetti]{peng_partitioning_2017}
Peng, R., Sun, H., and Zanetti, L.
\newblock Partitioning {Well}-{Clustered} {Graphs}: {Spectral} {Clustering} {Works}!
\newblock \emph{SIAM Journal on Computing}, 46\penalty0 (2):\penalty0 710--743, 2017.

\bibitem[Rand(1971)]{ari}
Rand, W.~M.
\newblock Objective criteria for the evaluation of clustering methods.
\newblock \emph{Journal of the American Statistical Association}, 66\penalty0 (336):\penalty0 846--850, 1971.

\bibitem[Saranurak \& Wang(2019)Saranurak and Wang]{SaranurakW19}
Saranurak, T. and Wang, D.
\newblock Expander decomposition and pruning: Faster, stronger, and simpler.
\newblock In \emph{\SODA{30th}{19}}, pp.\  2616--2635, 2019.

\bibitem[Sun \& Zanetti(2019)Sun and Zanetti]{SZ19}
Sun, H. and Zanetti, L.
\newblock Distributed graph clustering and sparsification.
\newblock \emph{{ACM} Transactions on Parallel Computing}, 6\penalty0 (3):\penalty0 17:1--17:23, 2019.

\bibitem[Thorup(2007)]{Thorup07}
Thorup, M.
\newblock Fully-dynamic min-cut.
\newblock \emph{Combinatorica}, 27\penalty0 (1):\penalty0 91--127, 2007.

\bibitem[Tropp(2012)]{tropp2012user}
Tropp, J.~A.
\newblock User-friendly tail bounds for sums of random matrices.
\newblock \emph{Foundations of Computational Mathematics}, 12\penalty0 (4):\penalty0 389--434, 2012.

\bibitem[von Luxburg(2007)]{sac/Luxburg07}
von Luxburg, U.
\newblock A tutorial on spectral clustering.
\newblock \emph{Statistics and Computing volume}, 17\penalty0 (4):\penalty0 395--416, 2007.

\end{thebibliography}
\bibliographystyle{icml2024}

\newpage
\appendix
\onecolumn
\section{Omitted Details from Section~\ref{sec:dynamic_CPS}}

The section presents the details omitted from Section~\ref{sec:dynamic_CPS}, and is organised as follows. We prove Lemma~\ref{lem:MS22+_GvsH} in Section~\ref{sec:ms22+_proof}, and prove Theorem~\ref{thm:dynamic_cp_sparsifier} in Section~\ref{sec:dynamic_cp_proof}.

\subsection{Proof of Lemma~\ref{lem:MS22+_GvsH}\label{sec:ms22+_proof}}

 We first prove a structure theorem. We   define the vectors $\chi_1, \ldots \chi_k$ to be the indicator vectors of the optimal clusters $S_1, \ldots, S_k$ in $G$, where $\chi_i(u) = 1$ if $u \in S_i$, and $\chi_i = 0$ otherwise. We further use $\bar{g}_1, \ldots, \bar{g}_k$ to denote the indicator vectors of the optimal clusters $S_1, \ldots S_k$ in $G$, normalised by the degrees in $H$, i.e.,
\begin{equation}\label{eq:normalised_indicator_vector}
\bar{g}_i \triangleq \frac{D_H^{\frac{1}{2}}\chi_i}{\|D_H^{\frac{1}{2}}\chi_i\|}.
\end{equation}
\begin{theorem}\label{thm:struc}
 Let $S_1, \ldots, S_k$ be a $k$-way partition of $G$ achieving $\rho_{G}(k)$, and   $\Upsilon_G(k) = \Omega(k)$, and $\{f_i\}_{i=1}^k$ be first $k$ eigenvectors of $\mathcal{L}_H$ and let  $\{\bar{g}_i\}_{i=1}^k$ be defined as in~\eqref{eq:normalised_indicator_vector} above. Then, the following statements hold:
\begin{enumerate}
    \item  For any $i\in[k]$, there is $\hatf_i\in\mathbb{R}^n$, which is  
    a linear combination of $f_1,\ldots, f_k$, such that $\|\barg_i - \hatf_i\|^2 = O \lp k/\Upsilon_G(k) \rp$ .
    \item  There are vectors $\hatg_1,\ldots, \hatg_k$, each of which is a linear combination of $\barg_1,\ldots, \barg_k$, such that $\sum_{i = 1}^k \norm{f_i - \hatg_i}^2 = O \lp k^2 /\Upsilon_G(k) \rp$.
\end{enumerate}
\end{theorem}
\begin{proof}
    Let $\hatf_i = \sum_{j=1}^k \langle \barg_i, f_j\rangle f_j $, and we write $\barg_i$ as a linear combination of the vectors $f_1, \ldots, f_n$ by 
   $
        \barg_i = \sum_{j = 1}^n \langle \barg_i, f_j \rangle f_j$.
   Since $\hatf_i$ is a projection of $\barg_i$, we have that $\barg_i - \hatf_i$ is perpendicular to $\hatf_i$ and 
    \begin{align*}
        \norm{\barg_i - \hatf_i}^2 & = \norm{\barg_i}^2 - \norm{\hatf_i}^2  = \left(\sum_{j = 1}^n \langle \barg_i, f_j \rangle^2 \right) - \left(\sum_{j = 1}^{k} \langle \barg_i, f_j \rangle^2 \right)   = \sum_{j = k + 1}^n \langle \barg_i, f_j \rangle^2.
    \end{align*}
    Now, let us consider the quadratic form
    \begin{align}
        \barg_i^\transpose \mathcal{L}_H \barg_i & = \left(\sum_{j = 1}^n \langle \barg_i, f_j \rangle f_j^\transpose \right) \mathcal{L}_H \left(\sum_{j = 1}^n \langle \barg_i, f_j \rangle f_j\right) \nonumber \\
        & = \sum_{j = 1}^n \langle \barg_i, f_j \rangle^2 \lambda_j(\mathcal{L}_H) \nonumber  \\
        & \geq \lambda_{k + 1}(\mathcal{L}_H)  \norm{\barg_i - \hatf_i}^2 \nonumber \\
        &= \Omega\lp \lambda_{k + 1}(\mathcal{L}_G) \rp \norm{\barg_i - \hatf_i}^2, \label{eq:lbquad}
    \end{align}
     where the second to last inequality follows by the fact that $\lambda_i (\mathcal{L}_H)\geq 0$ holds for any $1\leq i\leq n$, and the last inequality follows because $H$ is a cluster preserving sparsifier of $G$. This gives us that  
     \begin{align}
        \barg_i^\transpose \mathcal{L}_H \barg_i & = \sum_{(u, v) \in E_H} w_H(u, v) \left(\frac{\barg_i(u)}{\sqrt{\deg_H(u)}} - \frac{\barg_i(v)}{\sqrt{\deg_H(v)}}\right)^2 \nonumber \\
        & = \sum_{(u, v) \in E_H} w_H(u, v) \left(\frac{\chi_i(u)}{\sqrt{\vol_H(S_i)}} - \frac{\chi_i(v)}{\sqrt{\vol_H(S_i)}}\right)^2 \nonumber \\
        & = \frac{w_H(S_i, V \setminus S_i)}{\vol_H(S_j)}\nonumber\\
        & = O \lp k \cdot \rho_G(k) \rp,  \label{eq:upquad}
    \end{align}
    where the last line holds because $H$ is a cluster preserving sparsifier of $G$.
    Combining \eqref{eq:lbquad} with \eqref{eq:upquad}, we have that
    \[
    \norm{\barg_i - \hatf_i}^2 \leq \frac{\barg_i^\transpose \mathcal{L}_H \barg_i}{\Omega\lp \lambda_{k + 1}(\mathcal{L}_G) \rp } \leq \frac{O \lp k \cdot \rho_G(k) \rp}{\Omega\lp \lambda_{k + 1}(\mathcal{L}_G) \rp } = O \lp \frac{k}{\Upsilon_G (k)} \rp,
    \]
    which proves the first statement of the theorem.

 Now we prove the second statement.
We define for any $1\leq i \leq k$ that 
 $\hatg_i = \sum_{j=1}^k \langle f_i, \barg_j\rangle \barg_j$, and  have that 
    \begin{align*}
        \sum_{i = 1}^k \norm{f_i - \hatg_i}^2
        & = \sum_{i = 1}^k \left( \norm{f_i}^2 - \norm{\hatg_i}^2 \right) \\
        & = k - \sum_{i = 1}^k \sum_{j = 1}^k \langle \barg_j, f_i \rangle^2   \\
        & = \sum_{j = 1}^k \left( 1 - \sum_{i = 1}^k \langle \barg_j, f_i \rangle^2 \right) \\
        & = \sum_{j = 1}^k \left( \norm{\barg_j}^2 - \norm{\hatf_j}^2 \right) \\
        & = \sum_{j = 1}^k \norm{\barg_j - \hatf_j}^2\\
        & = \sum_{j = 1}^k O \lp \frac{k}{\Upsilon_G(k)} \rp\\
        & = O \lp \frac{k^2}{\Upsilon_G(k)} \rp,
    \end{align*}
    where the last inequality follows by the first statement of  Theorem~\ref{thm:struc}.
\end{proof}

\begin{proof}[Proof Sketch of Lemma~\ref{lem:MS22+_GvsH}] The proof follows Theorem~1.2 of \cite{peng_partitioning_2017} and Theorem~2 of \cite{MS22}, which imply that every returned cluster $P_i~(1\leq i\leq k)$ from spectral clustering on $G$ satisfies that
\[
\vol_G(P_i\triangle S_i) = O\left( k\cdot \frac{\vol_G(S_i)}{\Upsilon_G(k)}\right) 
\]
and
\[
\Phi_{G}(P_i) = O\left( \Phi_{G}(S_i) + \frac{k}{\Upsilon_G(k)} \right),
\]
where $S_i$ is the optimal correspondence of $P_i$ in $G$.  Since $H$ is a cluster-preserving sparsifier of $G$, we know that $\rho_H(k) = O(k\cdot \rho_G(k))$ and $\lambda_{k+1}(\mathcal{L}_H) = \Omega(\lambda_{k+1} (\mathcal{L}_G))$, which implies that 
\begin{equation}\label{eq:rhoh_vs_rhog}
\Upsilon_H(k) = \frac{\lambda_{k+1} (\mathcal{L}_H)}{\rho_H(k)} = \frac{\Omega(\lambda_{k+1} (\mathcal{L}_G))}{ O(k\cdot \rho_G(k))} =\Omega\left(\frac{1}{k}\cdot \Upsilon_G(k)\right).
\end{equation}
On the other side,   compared with their work, we need to apply the bottom $k$ eigenvectors of $\mathcal{L}_H$ instead of $\mathcal{L}_G$  to run spectral clustering. As such, combining~\eqref{eq:rhoh_vs_rhog} with the adjusted structure theorem~(Theorem~\ref{thm:struc}) one can prove Lemma~\ref{lem:MS22+_GvsH} using the proof technique from~\cite{MS22} and~\cite{peng_partitioning_2017}. 
\end{proof}

\subsection{Proof of Theorem~\ref{thm:dynamic_cp_sparsifier}\label{sec:dynamic_cp_proof}}

 Let 
\[
E_\mathrm{resampled} \triangleq  e~\bigcup~ \Big\{ \{u,v\}\in E_{G_{t+1}} ~|~ u\in V_\mathrm{doubled} \Big\} 
\]be the set of all the edges that have been (re)-sampled by Algorithm~\ref{alg:dynamic_update_cp_sparsifier}, and   
\[
E_\mathrm{old} \triangleq E_{t+1} \setminus E_\mathrm{resampled}.
\]
Moreover, let 
 $$
p^{(t+1)}_u(v)  \triangleq \min \left\{\frac{\tau \cdot \log (n_{t+1})}{\deg_{G_{t+1}}(u)} , 1\right\}
$$
be the ``ideal'' sampling probability of an edge $\{u,v\}$ if one runs 
the $\mathsf{SZ}$ algorithm from the scratch on  $G_{t+1}$, and let 
\[
q^{(t+1)}(u,v) \triangleq p^{(t+1)}_u(v) + p^{(t+1)}_v(u) - p^{(t+1)}_u(v) \cdot p^{(t+1)}_v(u)
\]be the probability that edge $e$ is sampled if one runs the $\mathsf{SZ}$ algorithm from scratch at time $t+1$.
 For any edge $\{u,v\}$, we use 
\begin{itemize}
    \item $\tilde{q}(u,v) \triangleq q^{(r)}(u,v) $
    \item $\tilde{p}_u(v) \triangleq p^{(r)}_u(v)$
    \item $\tilde{p}_v(u) \triangleq p^{(r)}_v(u)$
\end{itemize}for some $1 \leq r \leq t+1$ to denote the   sampling probability last  used for edge $\{u,v\}$ throughout the sequence of edge updates. Hence, we have $\tilde{q}(u,v) = q^{(t+1)}(u,v)$ if  $\{u,v\} \in E_\mathrm{resampled}$, and  $\tilde{q}(u,v) = q^{(r)}(u,v)$ for some $1 \leq r \leq t+1$ if edge $\{u,v\} \in E_\mathrm{old}$. By the algorithm description~(Line~\ref{alg:sz_dynamic:line:doubled} in Algorithm~\ref{alg:dynamic_update_cp_sparsifier}), we   know that 
 \begin{equation}\label{eq:SZrelation sampling probability}
    \frac{\tau \cdot \log(n_{t+1})}{2 \cdot \deg_{G_{t+1}}(u)} \leq \frac{\tau \cdot \log(n_r)}{\deg_{G_r}(u)} \leq \frac{2 \cdot \tau \cdot \log(n_{t+1})}{\deg_{G_{t+1}}(u)}.
\end{equation} The following two concentration inequalities will be used in our analysis.

\begin{lemma}[Bernstein's Inequality~\cite{chung2006concentration}]\label{lem:bernstein}
Let $X_1, \ldots, X_n$ be independent random variables such that $|X_i| \leq M$ for any $i \in \{1, \ldots, n\}$.
Let $X = \sum_{i = 1}^n X_i$, and  $R = \sum_{i = 1}^n \mathbb{E}\left[ X_i^2 \right]$.
Then, it holds that
\[
    \mathbb{P}\left[|X - \mathbb{E}\left[ X \right] | \geq t\right] \leq 2 \exp\left(-\frac{t^2}{2(R + Mt/3)}\right).
\]
\end{lemma}

\begin{lemma}[Matrix Chernoff Bound~\cite{tropp2012user}]\label{lem:matrix_chernoff}
    Consider a finite sequence $\{X_i\}$ of independent, random, PSD matrices of dimension $d$ that satisfy $\|X_i\| \leq R$.
    Let $\mu_{\mathrm{min}} \triangleq \lambda_{\mathrm{min}}(\mathbb{E}\left[ \sum_i X_i\right])$
    and $\mu_{\mathrm{max}} \triangleq \lambda_{\mathrm{max}}(\mathbb{E}\left[ \sum_i X_i\right])$.
    Then, it holds that
    \[
        \mathbb{P}\left[ \lambda_{\mathrm{min}}\left(\sum_i X_i\right) \leq (1 - \delta) \mu_{\mathrm{min}}  \right] \leq d \left( \frac{\mathrm{e}^{-\delta}}{(1 - \delta)^{1-\delta}}\right)^{\mu_{\mathrm{min}} / R}
    \]
    for $\delta \in [0, 1]$, and
    \[
        \mathbb{P}\left[ \lambda_{\mathrm{max}}\left(\sum_i X_i\right) \geq (1 + \delta) \mu_{\mathrm{max}}  \right] \leq d \left( \frac{\mathrm{e}^{\delta}}{(1 + \delta)^{1+\delta}}\right)^{\mu_{\mathrm{max}} / R}
    \]
    for $\delta \geq 0$.
\end{lemma}

We first prove the following result on the relationship of cut values between $G_{t+1}$ and $H_{t+1}$.

\begin{lemma}\label{lem:SZ_useful_concentration_bound}
   Let $G_{t+1}$ be a graph,  and $H_{t+1}$ the sparsifier returned by Algorithm~\ref{alg:dynamic_update_cp_sparsifier}. Suppose for every $\{u, v\} \in E_{t+1}$ that $\tilde{p}_u(v) < 1$, then it holds for any non-empty subset $A \subset V_{t+1}$ that
    \begin{align*}
    &\mathbb{P}\left[|w_{H_{t+1}}(A, V_{t+1} \setminus A) - w_{G_{t+1}}(A, V_{t+1} \setminus A)| \geq \frac{1}{2} \cdot w_{G_{t+1}}(A, V_{t+1} \setminus A)  \right] \\ 
    & \qquad \leq  2 \cdot \exp \left( \frac{- \tau \cdot \log n_{t+1} \cdot w_{G_{t+1}}(A, V_{t+1} \setminus A)}{10 \cdot \vol_{G_{t+1}}(A)}\right)
    \end{align*}
\end{lemma}
\begin{proof}
     For any edge $e=\{u,v\}$, we define the random variable $Y_e$ by
\begin{equation*}
Y_e \triangleq \begin{cases}
\frac{1}{\tilde{q}(u,v)} &\text{with probability $\tilde{q}(u,v)$}\\
0 &\text{otherwise.}
\end{cases}
\end{equation*}
We also define 
$$Z \triangleq \sum_{e \in E_{G_{t+1}}(A, V_{t+1} \setminus A)} Y_e,$$ and have that
\[
\ex{Z} = \sum_{e = \{u,v\} \in E_{G_{t+1}}(A, V_{t+1} \setminus A)} \ex{Y_e} = \sum_{e = \{u,v\} \in E_{G_{t+1}}(A, V_{t+1} \setminus A)} \tilde{q}(u,v) \cdot \tilde{q}(u,v)^{-1} = w_{G_{t+1}}(A, V_{t+1} \setminus A).
\]

 To prove a concentration bound on this degree estimate, we   apply the Bernstein inequality (Lemma~\ref{lem:bernstein}), for which we need to bound the second moment $$R \triangleq \sum_{e = \{u,v\} \in E_{G_{t+1}}(A, V_{t+1} \setminus A)} \ex{Y_e^2}.$$ We get that
 \begin{align}
    R &= \sum_{e = \{u,v\} \in E_{G_{t+1}}(A, V_{t+1} \setminus A)} \tilde{q}(u,v) \cdot \left(\frac{1}{\tilde{q}(u,v)}\right)^2 = \sum_{e = \{u,v\} \in E_{G_{t+1}}(A, V_{t+1} \setminus A)} \frac{1}{\tilde{q}(u,v)} \nonumber \\
     &\leq \sum_{e = \{u,v\} \in E_{G_{t+1}}(A, V_{t+1} \setminus A)} \frac{1}{\tilde{p}_u(v)} \label{deg_2nd_moment:eq1}\\
     &= \sum_{e = \{u,v\} \in E_{G_{t+1}}(A, V_{t+1} \setminus A)} \frac{2 \cdot \deg_{G_{t+1}}(u)}{\tau \cdot \log(n_{t+1})} \label{deg_2nd_moment:eq2} \\
     &\leq \frac{2 \cdot \Delta_{G_{t+1}}(A)}{ \tau \cdot \log(n_{t+1})} \cdot \sum_{e = \{u,v\} \in E_{G_{t+1}}(A, V_{t+1} \setminus A)}1   \nonumber \\
     &= \frac{2 \cdot \Delta_{G_{t+1}}(A) \cdot w_{G_{t+1}}(A, V_{t+1} \setminus A) }{\tau \cdot \log(n_{t+1})}, \nonumber
\end{align}where $\Delta_{G_{t+1}}(A) \triangleq \max_{u \in A} \deg_{G_{t+1}}(u)$, \eqref{deg_2nd_moment:eq1} holds since $\tilde{q}(u,v) = \tilde{p}_u(v) + \tilde{p}_v(u) - \tilde{p}_u(v) \cdot \tilde{p}_v(u) \geq \tilde{p}_u(v)$, and \eqref{deg_2nd_moment:eq2} holds because of~\eqref{eq:SZrelation sampling probability}.

 Note, by \eqref{eq:SZrelation sampling probability}, for any edge $e=\{u,v\} \in E_{G_{t+1}}(A, V_{t+1} \setminus A)$ we have that
\[
0 \leq Y_e = \frac{1}{\tilde{q}(u,v)} \leq \frac{1}{\tilde{p}_u(v)} \leq \frac{2 \cdot \Delta_{G_{t+1}}(A)}{\tau \cdot \log n_{t+1}}.
\]
 Then, by applying Bernstein's inequality, we have that
 \begin{align}
    \mathbb{P}\left[|Z - \ex{Z}| \geq \frac{1}{2} \ex{Z} \right] &\leq 2 \cdot \exp \left(-\frac{w_{G_{t+1}}(A, V_{t+1} \setminus A)^2/4}{\frac{\Delta_{G_{t+1}}(A) \cdot w_{G_{t+1}}(A, V_{t+1} \setminus A) }{\tau \cdot \log(n_{t+1})} + \frac{\Delta_{G_{t+1}}(A) \cdot w_{G_{t+1}}(A, V_{t+1} \setminus A) }{3 \cdot \tau \cdot \log(n_{t+1})}}\right)\\
    & = 2 \cdot \exp \left(-\frac{\tau \cdot \log(n_{t+1}) \cdot 3 \cdot w_{G_{t+1}}(A, V_{t+1} \setminus A)}{16 \cdot \Delta_{G_{t+1}}(A)} \right) \\
    &\leq 2 \cdot \exp \left(-\frac{\tau \cdot \log(n_{t+1}) \cdot w_{G_{t+1}}(A, V_{t+1} \setminus A)}{10 \cdot \vol_{G_{t+1}}(A)} \right), 
\end{align}
which proves the lemma.
\end{proof}

\begin{proof}[Proof of Theorem~\ref{thm:dynamic_cp_sparsifier}]
We first analyse the number of edges in $H_{t+1}$, i.e., the size of $F_{t+1}$. We have that
\[
\sum_{u \in V_{t+1}}\sum_{e=\{u,v\}\in E_{G_{t+1}}} \tilde{p}_u(v) \leq  \sum_{u \in V_{t+1}}\sum_{e=\{u,v\}\in E_{G_{t+1}}} \frac{2 \cdot \tau \cdot \log n_{t+1}}{\deg_{G_{t+1}}(u)} = 2 \cdot \tau \cdot n_{t+1}\cdot  \log n_{t+1}, 
\]
where the first inequality holds by \eqref{eq:SZrelation sampling probability}. Therefore, it holds by the Markov inequality that the number of edges $\{u,v\}$ with $\tilde{p}_u(v) \geq 1$ is $O \lp \tau \cdot n_{t+1} \log n_{t+1}\rp$. Without loss of generality, we assume that these edges are included in $F_{t+1}$, and we assume for the remaining part of the proof that it holds that $\tilde{p}_u(v) < 1$.

 We   now show that the degrees of the vertices in $G_{t+1}$ are approximately preserved in $H_{t+1}$. Let $u$ be an arbitrary vertex of $G_{t+1}$. Observing that $\vol_{G_{t+1}}(u) = w_{G_{t+1}}(u, V \setminus u) = \deg_{G_{t+1}}(u)$ and $w_{H_{t+1}}(u, V_{t+1} \setminus u) = \deg_{H_{t+1}}(u)$, by Lemma~\ref{lem:SZ_useful_concentration_bound} it holds that
\begin{align}
    \mathbb{P}\left[|\deg_{H_{t+1}}(u) - \deg_{G_{t+1}}(u)| \geq \frac{1}{2} \deg_{G_{t+1}}(u) \right]  
    &= 2 \exp \left(-(1/10) \cdot \tau \cdot \log n_{t+1}\right) \nonumber \\
    &= 2 \exp \left(-(1/10) \cdot (\log{n_{t+1}} \cdot C) / \lambda_{k+1}(\mathcal{L}_{G_{t+1}})\right) \nonumber \\
    &= o(1/n_{t+1}^2). \nonumber 
\end{align} Hence, by taking $C$ to be sufficiently large and   the union bound, it holds with high probability that the degrees of all the vertices   in $G_{t+1}$ are preserved in $H_{t+1}$ up to a constant factor. Throughout the rest of the proof, we   assume this is the case. This implies   for any subset $A \subseteq V_{t+1}$  that $\vol_{H_{t+1}}(A) = \Theta(\vol_{G_{t+1}}(A))$.

 Secondly, we prove   it holds that $\Phi_{H_{t+1}}(S_i) = O(k \cdot \Phi_{G_{t+1}}(S_i))$ for any $1 \leq i \leq k$, where $S_1, \ldots, S_k$  are the optimal clusters corresponding to $\rho_{G_{t+1}}(k)$. For any $1 \leq i \leq k$, it holds that
\[
\ex{w_{H_{t+1}}(S_i, V_{t+1} \setminus S_i)} = \sum_{\substack{e=\{u,v\} \in E_{t+1} \\ u \in S_i, v \notin S_i}} \tilde{q}(u,v) \cdot \frac{1}{\tilde{q}(u,v)} = w_{G_{t+1}}(S, V_{t+1}\setminus S_i).
\]
Hence, by Markov's inequality and the union bound, it holds with constant probability that $w_{H_{t+1}}(S_i, V_{t+1} \setminus S_i) = O(k \cdot w_{G_{t+1}}(S_i, V_{t+1} \setminus S_i))$. Therefore,  it holds with constant probability that 
\[
    \rho_{H_{t+1}}(k) \leq \max_{1 \leq i \leq k} \Phi_{H_{t+1}}(S_i) = \max_{1 \leq i \leq k} O(k \cdot \Phi_{G_{t+1}}(S_i)) = O(k \cdot \rho_{G_{t+1}}(k)).
\]

 Next, we prove that 
$\lambda_{k+1}(\mathcal{L}_{H_{t+1}}) = \Omega(\lambda_{k+1}(\mathcal{L}_{G_{t+1}}))$. Let $\overline{\mathcal{L}}_{G_{t+1}}$ be the projection of $\mathcal{L}_{G_{t+1}}$ on its top $n_{t+1} - k$ eigenspaces, and notice that $\overline{\mathcal{L}}_{G_{t+1}}$ can be written as
\[
    \overline{\mathcal{L}}_{G_{t+1}} = \sum_{i = k+1}^{n_{t+1}} \lambda_i(\mathcal{L}_{G_{t+1}})\cdot  f_i f_i^\intercal
\]
where $f_1, \ldots, f_{n_{t+1}}$ are the eigenvectors of $\mathcal{L}_{G_{t+1}}$.
Let $\overline{\mathcal{L}}_{G_{t+1}}^{-1/2}$ be the square root of the pseudoinverse of $\overline{\mathcal{L}}_{G_{t+1}}$. 
 We   prove that the top $n_{t+1} - k$ eigenvalues of $\mathcal{L}_{G_{t+1}}$ are preserved, which implies that $\lambda_{k+1}(\mathcal{L}_{H_{t+1}}) = \Theta(\lambda_{k+1}(\mathcal{L}_{G_{t+1}}))$.

 To prove this, for each edge $e = \{u,v\} \in E_{G_{t+1}}$ we define a random matrix $X_e \in \mathbb{R}^{n_{t+1} \times n_{t+1}}$ by
\begin{equation*}
X_e = \begin{cases}
w_{H_{t+1}}(u,v) \cdot \overline{\mathcal{L}}_{G_{t+1}}^{-1/2} b_e b_e^\intercal \overline{\mathcal{L}}_{G_{t+1}}^{-1/2} &\text{if $e=\{u,v\}$ is sampled by the algorithm}\\
0 &\text{otherwise,}
\end{cases}
\end{equation*}
 where $b_e \triangleq \chi_{u} - \chi_{v}$ is the  edge indicator vector and 
 $\chi_{v}\in\mathbb{R}^n$ is defined by
 \begin{equation*}
\chi_{v}(a) \triangleq  \begin{cases}
\frac{1}{\sqrt{\deg_{G_{t+1}}(v)}} &\text{if $a = v$ }\\
0 &\text{otherwise.}
\end{cases}
\end{equation*}
Notice that
\[
\sum_{e \in E_{G_{t+1}}} X_e = \sum_{\substack{ e=\{u,v\} \\ e \in E_{G_{t+1}}}} w_{H_{t+1}}(u,v) \cdot \overline{\mathcal{L}}_{G_{t+1}}^{-1/2} b_e b_e^\intercal \overline{\mathcal{L}}_{G_{t+1}}^{-1/2} = \overline{\mathcal{L}}_{G_{t+1}}^{-1/2} \mathcal{L}_{H'_{t+1}} \overline{\mathcal{L}}_{G_{t+1}}^{-1/2},
\]
where 
\[
\mathcal{L}_{H'_{t+1}} \triangleq \sum_{e \in E_{G_{t+1}}} w_{H_{t+1}}(u,v) \cdot  b_e b_e^\intercal 
\]
is $\mathcal{L}_{H_{t+1}}$ normalised with respect to the degree of the vertices in $G_{t+1}$. We prove that, with high probability, the top $n_{t+1}-k$ eigenvalues of $\mathcal{L}_{H'_{t+1}}$ and $\mathcal{L}_{G_{t+1}}$ are approximately the same. Then, to finish the proof, we also show that this is the case for the top $n_{t+1}-k$ eigenvalues of $\mathcal{L}_{H_{t+1}}$ and $\mathcal{L}_{H'_{t+1}}$, from which we get that $\lambda_{k+1}(\mathcal{L}_{H_{t+1}}) = \Omega \lp \lambda_{k+1}(\mathcal{L}_{G_{t+1}}) \rp$.

 First, from \eqref{eq:SZrelation sampling probability} we get that for any edge $e$ it holds that
\begin{equation}\label{eq:upper bound sample probability}
\tilde{q}(u,v) \leq \tilde{p}_u(v) + \tilde{p}_v(u) \leq 2 \cdot \lp \frac{\tau \cdot \log(n_{t+1})}{\deg_{G_{t+1}}(u)} + \frac{\tau \cdot \log(n_{t+1})}{ \deg_{G_{t+1}}(v)} \rp, 
\end{equation}
and 
\begin{equation}\label{eq:lower bound sample probability}
\tilde{q}(u,v) \geq \frac{1}{2}\cdot \lp\tilde{p}_u(v) + \tilde{p}_v(u)\rp \geq  \frac{1}{4}\cdot \lp \frac{\tau \cdot \log(n_{t+1})}{\deg_{G_{t+1}}(u)} + \frac{\tau \cdot \log(n_{t+1})}{ \deg_{G_{t+1}}(u)}\rp. 
\end{equation}  
 We start by calculating the first moment of $\sum_{e \in E_{G_{t+1}}} X_e$, and   have that
\begin{align*}
\mathbb{E}\left[\sum_{e \in E_{G_{t+1}}} X_e \right] &= \sum_{\substack{ e=\{u,v\}\\ e\in E_{G_{t+1}}}} \tilde{q}(u,v) \cdot w_{H_{t+1}}(u,v) \cdot \overline{\mathcal{L}}_{G_{t+1}}^{-1/2} b_e b_e^\intercal  \overline{\mathcal{L}}_{G_{t+1}}^{-1/2}\overline{\mathcal{L}}_{G_{t+1}}^{-1/2} \\
&= \sum_{\substack{ e=\{u,v\}\\ e\in E_{G_{t+1}}}} \tilde{q}(u,v) \cdot \frac{1}{\tilde{q}(u,v)} \cdot \overline{\mathcal{L}}_{G_{t+1}}^{-1/2} b_e b_e^\intercal \overline{\mathcal{L}}_{G_{t+1}}^{-1/2} \\ 
& =  \overline{\mathcal{L}}_{G_{t+1}}^{-1/2} \mathcal{L}_{G_{t+1}} \overline{\mathcal{L}}_{G_{t+1}}^{-1/2}. 
\end{align*}
Moreover, for any sampled $e = \{u,v\}$ we have that 
\begin{align}
    \| X_e \| &\leq w_{H_{t+1}}(u,v) \cdot b_e^\intercal \overline{\mathcal{L}}_{G_{t+1}}^{-1/2}  \overline{\mathcal{L}}_{G_{t+1}}^{-1/2}b_e \nonumber \\ 
     &=\frac{1}{\tilde{q}(u,v)} \cdot b_e^\intercal \overline{\mathcal{L}}_{G_{t+1}}^{-1/2}  \overline{\mathcal{L}}_{G_{t+1}}^{-1/2}b_e\nonumber \\ 
     &\leq \frac{1}{\tilde{q}(u,v)} \cdot \frac{1}{\lambda_{k+1}(\mathcal{L}_{G_{t+1}})} \cdot  \|b_e\|^2 \nonumber \\
     &\leq \frac{4 \lambda_{k+1}(\mathcal{L}_{G_{t+1}})}{C \cdot \log{n_{t+1}} \cdot \lp \frac{1}{\deg_{G_{t+1}}(u)} + \frac{1}{\deg_{G_{t+1}}(v)} \rp} \cdot \frac{1}{\lambda_{k+1}(\mathcal{L}_{G_{t+1}})} \cdot  \lp \frac{1}{\deg_{G_{t+1}}(u)} + \frac{1}{\deg_{G_{t+1}}(v)} \rp \label{eq1:matrix_norm_upperbound}\\
     &=\frac{4}{C \cdot \log{n_{t+1}}}, \nonumber
\end{align}
where the second inequality follows by the min-max theorem of eigenvalues, and~\eqref{eq1:matrix_norm_upperbound} holds by \eqref{eq:lower bound sample probability}. Now we apply the matrix Chernoff bound (Lemma~\ref{lem:matrix_chernoff}) to analyse the eigenvalues of $\sum_{e \in E_{G_{t+1}}}X_e$. We set   $\lambda_\mathrm{max}\left(\mathbb{E}\left[\sum_{e \in E_{G_{t+1}}} X_e \right]\right) = \lambda_\mathrm{max}\left(\overline{\mathcal{L}}_{G_{t+1}}^{-1/2} \mathcal{L}_{G_{t+1}} \overline{\mathcal{L}}_{G_{t+1}}^{-1/2}\right) = 1$, $R=\frac{4}{C \cdot \log{n_{t+1}}}$ and $\delta=1/2$, and   have that
\[
\mathbb{P}\left[\lambda_\mathrm{max}\lp\sum_{e \in E_{G_{t+1}}} X_e\rp \geq \frac{3}{2}\right] \leq n_{t+1} \cdot \lp \frac{e^{1/2}}{(1 + 1/2)^{3/2}}\rp^{C \cdot \log{n_{t+1}} / 4 } = O(1/n_{t+1}^c)
\]
for some constant $c$. Therefore we get that
\begin{equation}\label{eq:lambda_max_bound}
\mathbb{P}\left[\lambda_\mathrm{max}\lp\sum_{e \in E_{G_{t+1}}} X_e\rp < \frac{3}{2}\right] = 1 - O(1/n_{t+1}^c).
\end{equation}
Similarly, since $\lambda_\mathrm{min}\left(\mathbb{E}\left[\sum_{e \in E_{G_{t+1}}} X_e \right]\right) = \lambda_\mathrm{min}\left(\overline{\mathcal{L}}_{G_{t+1}}^{-1/2} \mathcal{L}_{G_{t+1}} \overline{\mathcal{L}}_{G_{t+1}}^{-1/2}\right) = 1$, the other side of the matrix Chernoff bound gives us that
\begin{equation}\label{eq:lambda_min_bound}
\mathbb{P}\left[\lambda_\mathrm{min}\lp\sum_{e \in E_{G_{t+1}}} X_e\rp > \frac{1}{2}\right] = 1 - O(1/n_{t+1}^c).
\end{equation}
Combining \eqref{eq:lambda_max_bound} and \eqref{eq:lambda_min_bound}, it holds with probability $1 - O(1/n_{t+1}^c)$ for any non-zero $x \in \mathbb{R}^{n_{t+1}}$ in the space spanned by $f_{k+1}, \ldots, f_{n_{t+1}}$ that
\[
    \frac{x^\intercal \overline{\mathcal{L}}_{G_{t+1}}^{-1/2} \mathcal{L}_{H_{t+1}}^{'} \overline{\mathcal{L}}_{G_{t+1}}^{-1/2}x}{x^\intercal x} \in (1/2, 3/2).
\]
 Since $\dmn(\spn\{f_{k+1}, \ldots, f_{n_{t+1}}\}) = n_{t+1}-k$,   there exist $n_{t+1}-k$ orthogonal vectors whose Rayleigh quotient with respect to $\mathcal{L}_{H_{t+1}}^{'}$ is $\Omega\lp\lambda_{k+1}\lp\mathcal{L}_{G_{t+1}}\rp\rp$. The Courant-Fischer Theorem implies that $\lambda_{k+1}(\mathcal{L}_{H_{t+1}}') = \Omega\lp\lambda_{k+1}\lp\mathcal{L}_{G_{t+1}}\rp\rp$.

 It only remains to show that $\lambda_{k+1}(\mathcal{L}_{H_{t+1}}) = \Omega(\lambda_{k+1}(\mathcal{L}_{H_{t+1}}'))$, which implies that $\lambda_{k+1}(\mathcal{L}_{H_{t+1}}) = \Omega\lp\lambda_{k+1}\lp\mathcal{L}_{G_{t+1}}\rp\rp$. By definition of $\lambda_{k+1}(\mathcal{L}_{H_{t+1}}')$, we have that 
\[
\mathcal{L}_{H_{t+1}} = D_{H_{t+1}}^{-1/2} D_{G_{t+1}}^{1/2} \mathcal{L}_{H_{t+1}}^{'} D_{G_{t+1}}^{1/2} D_{H_{t+1}}^{-1/2}.
\]
Therefore, for any $x \in \mathbb{R}^{n_{t+1}}$ and $y \triangleq D_{G_{t+1}}^{1/2} D_{H_{t+1}}^{-1/2} x$, it holds that
\[
    \frac{x^\intercal \mathcal{L}_{H_{t+1}} x}{x^\intercal x} = \frac{y^\intercal \mathcal{L}_{H_{t+1}}^{'} y}{x^\intercal x} = \Omega\left(\frac{y^\intercal \mathcal{L}_{H_{t+1}}^{'} y}{y^\intercal y}\right),
\]
where the final guarantee follows from the fact that the degrees in $H_{t+1}$ are preserved up to a constant factor.
The conclusion of the theorem follows from the Courant-Fischer Theorem.

Finally, it remains to analyse the amortised update time of the algorithm. Notice that, if one only needs to sample the incoming edge at time $t+1$,  then the update time is $O(1)$. Otherwise, all the edges adjacent to some vertex $w$ need to be resampled, and the running time for this step is $O(\deg_{G_{t+1}}(w))$. However, this means that either $\deg_{G_{t+1}}(w) > 2 \cdot \deg_{G_{t}}(w)$ or $\log (n_{t+1}) > 2 \cdot \log (n_t)$. In the first case, this only occurs at most every $\deg_{G_{t}}(w)$ edge updates, which results in the amortised update  time of $O(1)$. The second case  only happens after every $n_t^2$ vertex additions, and   in the worst case we only have to resample all the edges in present in $G_t$ every $n_t^2$ edge updates, which again leads to the amortised update time of  $O(1)$.
\end{proof}

\section{Omitted Details from Section~\ref{sec:algorithm}}
 This section contains the omitted details from Section~\ref{sec:algorithm}, and is organised as follows. In Section~\ref{sec:algorithm:notation} we introduce additional notation to analyse our constructed contracted graphs. In Section~\ref{sec:algorithm:proofs} we present the omitted proofs for Lemmas~\ref{lem:time_complexity_contract_graph},~\ref{lem:running_time_dynamic_update_contracted_graph},~\ref{lem:conductanceHvsG},~\ref{lem:conductance_full>contract}, and~\ref{lem:lambda_k+1 from tilde(G) to G}, and we formally describe the $\mathsf{UpdateContractedGraph}$ procedure.

\subsection{Notation}\label{sec:algorithm:notation}

 For any subset $A \subset V_{t'}$,   let $
\tilde{A} \triangleq A \cap \tilde{V}^{\mathrm{nc}}_{t'}$ 
be the representation of $A$ among the non-contracted vertices of $\tilde{G}_{t'}$. Recall   that for any subset of vertices $A \subset V_{t'}$,  we use $
A^{(t)} \triangleq A \cap V_t$
to denote the set of  vertices present at time $t$. 
Let 
\[
E_\mathrm{added} \triangleq \Enew \cup \Big\{\{u,v\} \in E_t~\mid~\deg_{G_{t'}}(u) > 2 \cdot \deg_{G_r}(u)  \text{ or }\deg_{G_{t'}}(v) > 2 \cdot \deg_{G_r}(v)\Big\}
\]
be the set of edges that have been directly added into $\tilde{G}_t$, where $\deg_{G_r}(w)$ for $r \leq t$ is the degree of $w$ used to construct the contracted graph. These edges are the ones directly added as new edges or their endpoints are pulled out from   clusters in $\widetilde{G}_t$.
 For a subset $B \subset \tilde{V}_{t'}$, let $\hat{B}$ be the representation of the set $B$ in   $G_{t'}$, i.e.,
    $$\hat{B} \triangleq B^\mathrm{nc}  \bigcup\left( \bigcup_{p_i \in B^\mathrm{c}}P^{(t')}_i\right),$$
    where $P^{(t')}_i \triangleq P_i \setminus (P_i \cap \tilde{V}^{\mathrm{nc}}_{t'})$, $B^\mathrm{nc} \triangleq B \cap V^{\mathrm{nc}}_{t'}$, and $B^\mathrm{c} \triangleq B \cap V^{\mathrm{c}}_{t'}$. One can see $P^{(t')}_i$ as the vertices in $P_i$ that are still represented by the respective super vertex in $\widetilde{G}$.

\subsection{Omitted Proofs}\label{sec:algorithm:proofs}

Our analysis is based on approximation guarantee of spectral clustering. The following result,   which can be  shown easily by combining the proof technique of \cite{peng_partitioning_2017} and the one of \cite{MS22}, will be used in our analysis.

\begin{lemma}\label{lem:MS22+}
There is an absolute  constant  $\CGap\in\mathbb{R}_{>0}$, such that the following holds: 
Let $G$ be a graph with $k$ optimal clusters $\{S_i\}_{i=1}^k$, and $\Upsilon_G(k) \geq \CGap \cdot k$. Let $\{P_i\}_{i=1}^k$ be the output of spectral clustering and, without loss of generality, the optimal correspondence of $P_i$ is $S_i$ for any $1\leq i \leq k$. Then, it holds  for any $1\leq i\leq k$ that 
\[
  \vol_G(P_i\triangle S_i)  
  \leq \frac{k \cdot \CGap}{3\Upsilon_G(k)} \cdot \vol_G(S_i),
\]
where $A\triangle B$ for any sets $A$ and $B$ is defined by $A\triangle B\triangleq (A\setminus B)\cup (B\setminus A)$.
  It also holds that
 \[
 \Phi_{G}(P_i) = O\lp \Phi_{G}(S_i) + \frac{k}{\Upsilon_G(k)} \rp.
 \]
 Moreover, these $P_1,\ldots, P_k$ can be computed in nearly-linear time.
\end{lemma}

\begin{algorithm}
    \caption{$\mathsf{UpdateContractedGraph}(G_t, \tilde{G}_t, e)$}\label{alg:update_contracted_graph}
\begin{algorithmic}[1]
    
    \STATE \textbf{Input:} Graph $G_t=(V_t, E_t)$, contracted graph $\tilde{G}_t=(\tilde{V}_t ,\tilde{E}_t, w_{\tilde{G}_{t}})$, incoming edge $e = \{u,v\}$.
    
    \STATE \textbf{Output:} Contracted graph $\tilde{G}_{t+1}=(\tilde{V}_{t+1}, \tilde{E}_{t+1}, w_{\tilde{G}_{t+1}}) $

    \STATE $\Vnew \leftarrow \{u,v\} \setminus V_t$ 
    
    \STATE $G_{t+1} \leftarrow (V_t \cup \Vnew, E_t \cup e)$

    \STATE $\tilde{G}_{t+1} \leftarrow (\tilde{V}_t \cup \Vnew, \tilde{E}_t, w_{\tilde{G}_t}) = (\tilde{V}_{t+1}, \tilde{E}_{t+1}, w_{\tilde{G}_{t+1}})$

    \STATE $\tilde{V}_{t+1}^{\mathrm{nc}} \leftarrow \tilde{V}_{t+1}^{\mathrm{nc}} \cup \Vnew$

    \FOR{$w \in \{u,v\} \setminus V_\mathrm{new}$} \label{algupdatecontractgraph:line:start_add_vertices}
        
        \STATE Let $G_r$ be the graph at time $r$ when the contracted graph is constructed, and $H_r = (V_r, F_r, w_{H_r})$ the cluster preserving sparsifier at time $r$.
    
        \IF{$w \notin \tilde{V}_{t+1}^\mathrm{nc}$ and $\deg_{G_{t+1}}(w) > 2 \cdot \deg_{G_{r}}(w)$}

        \STATE Let $p_j$ be the super node such that $w \in P_j$

        \STATE $\tilde{V}_{t+1}^\mathrm{nc} \leftarrow \tilde{V}_{t+1}^\mathrm{nc} \cup w$

        \STATE $\tilde{E}_{t+1} \leftarrow \tilde{E}_{t+1} \cup E_{G_{t+1}}(w, \tilde{V}_{t+1}^\mathrm{nc})$

        \FOR{$\hat{v} \in \tilde{V}^\mathrm{nc}_{t+1}$ adjacent to $w$}{
        \STATE $w_{\tilde{G}_{t+1}}(p_j,  \hat{v}) \leftarrow w_{\tilde{G}_{t+1}}(p_j,  \hat{v}) - 1$
        }
        \ENDFOR

        \FOR{$\{w,p_i\} \in w \times \tilde{V}_{t+1}^\mathrm{c}$}
        
        \STATE 
        $\tilde{E}_{t+1} \leftarrow \tilde{E}_{t+1} \cup \{w,p_i\}$
        
        \STATE $w_{\tilde{G}_{t+1}}(w,p_i) \leftarrow w_{G_{t+1}}(w, P^{(t+1)}_i)$  

        \STATE $w_{\tilde{G}_{t+1}}(p_i,p_j) \leftarrow w_{\tilde{G}_{t+1}}(p_i,p_j) - w_{H_{r}}(w, P^{(t+1)}_i)$ 
    \ENDFOR
    
    \ENDIF
    \ENDFOR
    \label{algupdatecontractgraph:line:end_add_vertices}

    \IF{$u \in \tilde{V}_{t+1}^\mathrm{nc}$ and $v \in \tilde{V}_{t+1}^\mathrm{nc}$}
    {
   
    \STATE $\tilde{E}_{t+1} \leftarrow \tilde{E}_{t+1} \cup \{u,v\}$
    }
    \ELSIF{$u \in \tilde{V}_{t+1}^\mathrm{nc}$ or $v \in \tilde{V}_{t+1}^\mathrm{nc}$}{
    
    \STATE Without loss of generality, let $u \in \tilde{V}_{t+1}^\mathrm{nc}$ and $v \notin \tilde{V}_{t+1}^\mathrm{nc}$.
    Let $p_j$ be the supernode such that $v \in P_j$

    \STATE $\tilde{E}_{t+1} \leftarrow \tilde{E}_{t+1} \cup \{u,p_j\}$

    \STATE $w_{\tilde{G}_{t+1}}(u,p_j) \leftarrow w_{\tilde{G}_{t+1}}(u,p_j) + 1$
    }
    \ELSE{
    \STATE Let $p_i$ and $p_j$ be the supernodes such that $u \in P_i$ and $v \in P_j$
    
    \STATE $w_{\tilde{G}_{t+1}}(p_i,p_j) \leftarrow w_{\tilde{G}_{t+1}}(p_i,p_j) + 1$
    }
\ENDIF
    \STATE \textbf{Return} {$G_{t+1}=(\tilde{V}_{t+1} ,\tilde{E}_{t+1}, w_{\tilde{G}_{t+1}})$}

\end{algorithmic}
\end{algorithm}

\begin{proof}[Proof of Lemma~\ref{lem:time_complexity_contract_graph}]
 The running time of the algorithm is dominated by computing the total weight $w_{H_t}(P_i, P_j)$ between every $P_i, P_j \in \mathcal{P}$ (Lines~\ref{algcontractgraph:line:start_add_edges_big_clusters}--\ref{algcontractgraph:line:end_add_edges_big_clusters}), which takes $O(|F_t|)$ time as there are $|F_t|$ edges in $H_t$. 
\end{proof}

\begin{proof}[Proof of Lemma~\ref{lem:running_time_dynamic_update_contracted_graph}]
    The running time of the update operation is dominated by the case in which  a vertex is  pulled out from a contracted vertex~(Lines~\ref{algupdatecontractgraph:line:start_add_vertices}--\ref{algupdatecontractgraph:line:end_add_vertices}). It's easy to see that, if this does not happen, then the running time is $O(1)$  as the edge is just added into $\tilde{G}_t$. 
    
    Let $\{u,v\}$ be the added edge, and we assume   Lines~\ref{algupdatecontractgraph:line:start_add_vertices}--\ref{algupdatecontractgraph:line:end_add_vertices} are triggered. The running time for this case is  $O(\deg_{G_{t+1}}(u) + \deg_{G_{t+1}}(v))$, since at least one of $u$ and $v$ is pulled out from their respective contracted vertices and all the adjacent edges are placed into the contracted graph. Notice that this only happens if  $\deg_{G_{t+1}}(u) > 2\cdot \deg_{G_{t}}(u)$ or $\deg_{G_{t+1}}(v) > 2\cdot \deg_{G_{t}}(v)$. Since at least $\deg_{G_{t}}(u)$ or $\deg_{G_{t}}(v)$ edge insertions are needed before running Lines~\ref{algupdatecontractgraph:line:start_add_vertices}--\ref{algupdatecontractgraph:line:end_add_vertices},   the amortised per edge update time is $O(1)$.  
\end{proof}

\begin{proof}[Proof of Lemma~\ref{lem:conductanceHvsG}]
     Notice by Lemma~\ref{lem:MS22+} we know it holds with high probability for all $1 \leq i \leq k$ that 
    $ 
    \Phi_{H_t}(P_i) = O \lp k \cdot \rho_{H_t}(k) \rp$.
    By applying Theorem~\ref{thm:dynamic_cp_sparsifier}, it holds with high probability that 
     $ \Phi_{H_t}(P_i) =  O \lp k^2 \cdot \rho_{G_t}(k) \rp$.
    By Lemma~\ref{lem:MS22+_GvsH}, we also have  with high probability   that
      $
        \Phi_{G_t}(P_i) =  O \lp k^2 \cdot \rho_{G_t}(k) \rp$.
    This proves the statement.
\end{proof}

The next lemma shows that, starting from $G_t$ and $H_t$, one can easily construct a cluster preserving sparsifier of $G_{t'}$.

\begin{lemma}\label{lem:special_sparsifier}
    Let $H'_{t'} \triangleq (V_{t'}, F_t \cup E_\mathrm{added}, w_{H'_{t'}})$ be a graph, where
    \begin{equation*}
w_{H'_{t'}}(e) \triangleq  \begin{cases}
1 &\text{$e \in E_\mathrm{added}$}\\
w_{H_t}(e) &\text{$e \in F_t \setminus E_\mathrm{added}$} \\
0 &\text{otherwise.}
\end{cases}
\end{equation*}
Then, it holds   with high probability that $H'_{t'}$ is a cluster preserving sparsifier of $G_{t'}$.  
\end{lemma}

\begin{proof} 
    First, for any $e\in E_\mathrm{added}$ we know that it is included in $H'_{t'}$ with probability 1.   For any other edge $e=\{u,v\} \in F_t \setminus E_\mathrm{added}$, we know by the construction of $H_t$ using the dynamic cluster-preserving sparsifier that the   parameter used to sample $e$ from the perspective of $u$ is  
     \begin{equation*}
    \frac{\tau \cdot \log(n_{t})}{2 \cdot \deg_{G_{t}}(u)} \leq \frac{\tau \cdot \log(n_r)}{\deg_{G_r}(u)} \leq \frac{2 \cdot \tau \cdot \log(n_{t})}{\deg_{G_{t}}(u)},
\end{equation*} 
for some $1 \leq r \leq t$. We also know by construction that for any $e=\{u,v\} \in F_t \setminus E_\mathrm{added}$ that
\[
\deg_{G_{t'}}(u) \leq 2 \cdot \deg_{G_t}(u).
\]
Finally, it holds that $\log(n_t) \leq \log(n_{t'}) \leq  2 \log(n_{t})$. From this we get that $e$ is sampled from vertex $u$ with the following parameter
\begin{equation*}
    \frac{\tau \cdot \log(n_{t'})}{4 \cdot \deg_{G_{t'}}(u)} \leq \frac{\tau \cdot \log(n_r)}{\deg_{G_r}(u)} \leq \frac{4 \cdot \tau \cdot \log(n_{t'})}{\deg_{G_{t'}}(u)}.
\end{equation*} 
Following almost the same analysis as the proof of Theorem~\ref{thm:dynamic_cp_sparsifier}, it holds with high probability that $H'_{t'}$ is a cluster preserving sparsifier of $G_{t'}$.   
\end{proof}

 Our next lemma proves several useful properties about the contracted graph as it is updated.

\begin{lemma}\label{lem:new_cluster_properties}
    The following statements hold:
    \begin{enumerate}[label=(C\arabic*)]
        \item It holds for any subset $B \subset V_{t'} \setminus \tilde{V}_{t'}^\mathrm{nc}$ that  $\vol_{G_{t'}}(B) \leq 2 \cdot \vol_{G_{t}}(B)$.
        \item Suppose for a subset $A \subset V_{t'}$ with $\vol_{G_{t'}}(A) \leq  \vol({G_{t'}})/2$ we have that $\Phi_{G_t}(A^{(t)}) \geq 1/c_1$ and $\Phi_{G_{t'}}(A) \leq \log^{-\varepsilon} (n_{t'})$ for any positive  $c_1, \epsilon$ such that $4 \cdot c_1 \leq \log^{\varepsilon}(n_{t'})$, then it holds that 
        $$\Phi_{\tilde{G}_{t'}}(\tilde{A}) \leq \frac{21 \cdot c_1}{ \log^{\varepsilon}(n_{t'})}.$$
        \item For any super node $p_i \in \tilde{V}_t^\mathrm{c}$, it holds that 
         \[
        \Phi_{\tilde{G}_{t}}(p_i) =  O \lp k^{-6} \cdot \log^{-2\gamma}(n_t) \rp,
        \]
        and 
        \[
        \Phi_{\tilde{G}_{t'}}(p_i) =  O \lp k^{-6} \cdot \log^{-\gamma}(n_t) \rp.
        \]
    \end{enumerate}
\end{lemma}

Informally speaking, Property~(C1) of Lemma~\ref{lem:new_cluster_properties} shows that the volume of any vertex set $B\subset V_{t'} $ that are not directly represented in $\widetilde{G}_{t'}$ remains approximately the same in $G_t$ and $G_{t'}$; Property~(C2) states that, if the conductance of any set $A\subset V_{t'}$ in $G_{t'}$ becomes much lower than the one in $G_t$, then its representative set $\widetilde{A}\subset \widetilde{V}_{t'}$ has low conductance;  Property~(C3) further shows that the conductance of all the contracted vertices doesn't change significantly over time.

\begin{proof}[Proof of Lemma~\ref{lem:new_cluster_properties}]
 For (C1), by construction we have that for any $u \in V_{t'} \setminus \tilde{V}_{t'}^\mathrm{nc}$ it holds that $\deg_{G_{t'}}(u) \leq 2 \cdot \deg_{G_{t}}(u)$, from which the statement follows.

Next, we prove (C2). The following two claims will be used in our analysis. 
    
    \begin{claim} \label{claim:added_volume_needed}
    It holds that $\vol_{\Enew}(A) \geq \frac{\vol_{G_t}(A^{(t)}) \cdot \log^{\varepsilon}(n_{t'})}{2 \cdot c_1}$.
    \end{claim}
    \begin{proof}
         Assume by contradiction that $\vol_{\Enew}(A) < \frac{\vol_{G_t}(A^{(t)}) \cdot \log^{\varepsilon}(n_{t'})}{2 \cdot c_1}$. We have that
 \begin{align}
        \Phi_{G_{t'}}(A) &= \frac{w_{G_t}(A^{(t)}, V_t \setminus A^{(t)}) + w_{\Enew}(A, V_t \setminus A)}{\vol_{G_t}(A^{(t)}) + \vol_{\Enew}(A)} \nonumber\\
        &\geq \frac{w_{G_t}(A^{(t)}, V_t \setminus A^{(t)})}{\vol_{G_t}(A^{(t)}) + \vol_{\Enew}(A)} \nonumber\\
        &\geq \frac{1}{2} \min \left\{\frac{w_{G_t}(A^{(t)}, V_t \setminus A^{(t)})}{\vol_{G_t}(A^{(t)})}, \frac{w_{G_t}(A^{(t)}, V_t \setminus A^{(t)})}{\vol_{\Enew}(A)} \right\} \nonumber\\
        &\geq \frac{1}{2} \min \left\{\Phi_{G_t}(A^{(t)}), \frac{\vol_{G_t}(A^{(t)})}{c_1 \cdot   \vol_{\Enew}(A)} \right\} \nonumber \\ 
        &> \frac{1}{\log^{\varepsilon}(n_{t'})}, \nonumber
    \end{align}
    where on the last line we used the contradictory assumption. This contradicts the condition of $\Phi_{G_{t'}}(A)\leq \log^{-\varepsilon}(n_{t'})$, and hence the statement holds. 
    \end{proof}    
    
    Notice that this claim implies that
    \begin{align} 
        \vol_{G_{t'}}(A) &= \vol_{G_{t}}(A^{(t)}) + \vol_{\Enew}(A) \nonumber \\
        &\geq \left(1 + \frac{\log^{\varepsilon}(n_{t'})}{2 \cdot c_1}\right) \cdot \vol_{G_{t}}(A^{(t)}) \nonumber \\
        &\geq \left(1 + \frac{\log^{\varepsilon}(n_{t'})}{4 \cdot c_1}  + \frac{\log^{\varepsilon}(n_{t'})}{4 \cdot c_1}\right) \cdot \vol_{G_{t}}(A^{(t)}) \nonumber \\
        &\geq \left(2 + \frac{\log^{\varepsilon}(n_{t'})}{4 \cdot c_1}\right) \cdot \vol_{G_{t}}(A^{(t)}) \label{eq:lower_bound_volGt'}
    \end{align}
    where the last inequality follows from the fact that $4 \cdot c_1 \leq \log^\varepsilon(n_{t'})$.

    \begin{claim} \label{claim:volume_in_contract_graph}
     It holds that $\vol_{\tilde{G}_{t'}}(\tilde{A}) \geq \frac{\log^\varepsilon(n_{t'})}{4 \cdot c_1}\cdot \vol_{G_t}(A^{(t)})$.
    \end{claim}
    \begin{proof}
        Assume by contradiction that $\vol_{\tilde{G}_{t'}}(\tilde{A}) < \frac{\log^\varepsilon(n_{t'})}{4 \cdot c_1}\cdot \vol_{G_t}(A^{(t)})$. Then, it holds that
    \begin{align*}
        \vol_{G_{t'}}(A) &= \vol_{\tilde{G}_{t'}}(\tilde{A}) + \vol_{G_{t'}}(A \setminus \tilde{A})  \\
        &\leq \vol_{\tilde{G}_{t'}}(\tilde{A}) + 2 \cdot \vol_{G_{t}}(A \setminus \tilde{A}) \\
        &<  \frac{\log^\varepsilon(n_{t'})}{4 \cdot c_1}\cdot \vol_{G_t}(A^{(t)}) + 2 \cdot \vol_{G_{t}}(A^{(t)}),
    \end{align*}
    where the first inequality holds by statement (C1). Hence,  we reach a contradiction with \eqref{eq:lower_bound_volGt'}, and the claim holds.
    \end{proof}

    Now we are ready to prove statement (C2). We have that  
     \begin{align}
        \Phi_{\tilde{G}_{t'}}(\tilde{A}) &= \frac{w_{\tilde{G}_{t'}}(\tilde{A}, \tilde{V}_{t'} \setminus \tilde{A})}{\vol_{\tilde{G}_{t'}}(\tilde{A})} \nonumber \\
        &\leq \frac{w_{G_{t'}}(A, V_t \setminus A) + w_{G_{t'}}(A \setminus \tilde{A}, \tilde{A})}{\vol_{\tilde{G}_{t'}}(\tilde{A})} \nonumber \\
        &\leq \frac{\log^{-\varepsilon}(n_{t'}) \cdot \vol_{G_{t'}}(A) + \vol_{G_{t'}}(A \setminus \tilde{A})}{\vol_{\tilde{G}_{t'}}(\tilde{A})} \nonumber \\ 
        &\leq \frac{ \vol_{G_{t'}}(A)}{\log^\varepsilon(n_{t'}) \cdot \vol_{\tilde{G}_{t'}}(\tilde{A})} + \frac{8 \cdot c_1 \cdot \vol_{G_{t}}(A^{(t)})}{\vol_{G_t}(A^{(t)}) \cdot \log^\varepsilon(n_{t'})} \label{eq:line1:conductance_new_cluster}\\
         &\leq \frac{ \vol_{G_{t}}(A^{(t)}) + \vol_{\Enew}(A)}{\log^\varepsilon(n_{t'}) \cdot \vol_{\tilde{G}_{t'}}(\tilde{A})} + \frac{8 \cdot c_1}{ \log^\varepsilon(n_{t'})}  \nonumber \\
         &\leq \frac{3 \cdot \vol_{G_{t}}(A^{(t)}) + \vol_{\tilde{G}_{t'}}(\tilde{A})}{\log^\varepsilon(n_{t'}) \cdot \vol_{\tilde{G}_{t'}}(\tilde{A})} + \frac{8 \cdot c_1}{ \log^\varepsilon(n_{t'})}  \label{eq:line2:conductance_new_cluster} \\
         &\leq \frac{3 \cdot \vol_{G_{t}}(A^{(t)}) \cdot c_1 \cdot 4}{\log^{2\varepsilon}(n_{t'}) \cdot  \vol_{G_{t}}(A^{(t)})} + \frac{\vol_{\tilde{G}_{t'}}(A)} {\log^\varepsilon(n_{t'}) \cdot \vol_{\tilde{G}_{t'}}(\tilde{A})} + \frac{8 \cdot c_1}{ \log^\varepsilon(n_{t'})} \label{eq:line3:conductance_new_cluster}\\ 
         &\leq \frac{ 12 \cdot c_1}{\log^{2\varepsilon}(n_{t'})} + \frac{1} {\log^\varepsilon(n_{t'}) } + \frac{8 \cdot c_1}{ \log^\varepsilon(n_{t'})} \nonumber \\ 
         &\leq \frac{1 + 20 \cdot c_1}{\log^{\varepsilon}(n_{t'})} \leq \frac{21 \cdot c_1}{\log^{\varepsilon}(n_{t'})}, \nonumber
    \end{align}
   where \eqref{eq:line1:conductance_new_cluster} holds by Claim~\ref{claim:volume_in_contract_graph} and the fact that $\vol_{G_{t'}}(A \setminus \tilde{A}) \leq 2 \cdot \vol_{G_{t}}(A \setminus \tilde{A}) \leq 2 \cdot \vol_{G_t}(A^{(t)})$ by statement (C1).
    \eqref{eq:line2:conductance_new_cluster} holds because by construction $\vol_{\Enew}(A) \leq \vol_{\tilde{G}_{t'}}(\tilde{A}) + \vol_{G_{t'}}(A \setminus \tilde{A}) \leq \vol_{\tilde{G}_{t'}}(\tilde{A}) + 2\cdot \vol_{G_t}(A^{(t)})$, and
    \eqref{eq:line3:conductance_new_cluster} holds because of Claim~\ref{claim:volume_in_contract_graph}.

     Finally, we prove statement (C3).  For any $P_i \in \mathcal{P}$, we have by construction that
      \begin{align}
         &\Phi_{\tilde{G}_t}(p_i)  = \Phi_{H_t}(P_i)  = O \lp k^{-6} \cdot \log^{-2\gamma}(n_t) \rp, \label{eq2:contract_graph_conductance_large_cluster} 
     \end{align}
      where the last equality holds by Lemma~\ref{lem:conductanceHvsG}. This proves the first part of the statement. Next, notice that for any $p_i \in \tilde{V}_t^\mathrm{c}$, because $G_t$ is connected and each $P_i$ has almost identical volume as the corresponding optimal $S_i$ in $G_t$ (Lemma~\ref{lem:MS22+_GvsH}), by construction it holds that
 \begin{equation}\label{eq:vol_supernode_lowerbound_1}
         \vol_{\tilde{G}_t}(p_i) = \Omega(k^6 \cdot \log^{2\gamma}(n_t)),
     \end{equation}
     and  
\begin{equation}\label{eq:vol_supernode_lowerbound_2}
         \vol_{\tilde{G}_t}(\tilde{V}_t \setminus p_i)  = \Omega(k^6 \cdot \log^{2\gamma}(n_t)).
     \end{equation}
     Taking this into account, we get that
     \begin{align}
         w_{\tilde{G}_{t'}}(p_i, \tilde{V}_{t'} \setminus p_i) &\leq w_{\tilde{G}_{t}}(p_i, \tilde{V}_t \setminus p_i) + |\Enew| + w_{G_{t'}}\lp P_i \cap \tilde{V}_{t'}^\mathrm{nc} , P_i \setminus (P_i \cap \tilde{V}_{t'}^\mathrm{nc})\rp \nonumber\\
         &\leq w_{\tilde{G}_{t}}(p_i, \tilde{V} \setminus p_i) + \log^{\gamma}(n_t) + \vol_{G_{t'}}\lp P_i \cap \tilde{V}_{t'}^\mathrm{nc}\rp \nonumber \\
         &\leq \Phi_{\tilde{G}_t}(p_i) \cdot \min\{\vol_{\tilde{G}_t}(p_i), \vol_{\tilde{G}_t}(\tilde{V} \setminus p_i)\} + \log^{\gamma}(n_t) + 2 \cdot \log^{\gamma}(n_t) \label{eq2:upperbound_supernode_weight}\\
         &= O \lp k^{-6} \cdot \log^{-2\gamma}(n_t) \rp \cdot \min\{\vol_{\tilde{G}_t}(p_i), \vol_{\tilde{G}_t}(\tilde{V} \setminus p_i)\} + 3 \cdot \log^{\gamma}(n_t), \label{eq2.5:upperbound_supernode_weight}
     \end{align} 
     where \eqref{eq2:upperbound_supernode_weight} holds because  $\vol_{G_{t'}} ( P_i \cap \tilde{V}_{t'}^\mathrm{nc} ) \leq 2 \cdot \log^{\gamma}(n_t)$ as every vertex that is pulled out of $p_i$ needs to at least double in degree, so   adding $|\Enew|$ edges ensures at most $2\cdot|\Enew|$ volume can be pulled out of $p_i$,  \eqref{eq2.5:upperbound_supernode_weight} holds because of \eqref{eq2:contract_graph_conductance_large_cluster}.
     Moreover, we also have that
      \begin{align}
         \min\{\vol_{\tilde{G}_{t'}}(p_i), \vol_{\tilde{G}_{t'}}(\tilde{V}_{t'} \setminus p_i)\} &\geq \min\{\vol_{\tilde{G}_{t}}(p_i) - 2 \cdot \log^{\gamma}(n_t), \vol_{\tilde{G}_{t}}(\tilde{V}_{t} \setminus p_i)\} \label{eq1:vol_supernode_lowerbound_3}\\
         &= \Omega \lp \min\{\vol_{\tilde{G}_{t}}(p_i), \vol_{\tilde{G}_{t}}(\tilde{V}_{t} \setminus p_i)\} \rp \label{eq2:vol_supernode_lowerbound_3},
     \end{align}
     where \eqref{eq1:vol_supernode_lowerbound_3} holds because  $\vol_{\tilde{G}_{t'}}(p_i) \geq \vol_{\tilde{G}_{t}}(p_i) - \vol_{G_{t'}}( P_i \cap \tilde{V}_{t'}^\mathrm{nc}) \geq \vol_{\tilde{G}_{t'}}(p_i) - 2 \cdot \log^{\gamma}(n_t)$,  and \eqref{eq2:vol_supernode_lowerbound_3} holds because of \eqref{eq:vol_supernode_lowerbound_1} and \eqref{eq:vol_supernode_lowerbound_2}.
    Combining \eqref{eq2.5:upperbound_supernode_weight} and \eqref{eq2:vol_supernode_lowerbound_3}, we have   for any  $p_i \in \tilde{V}_t^\mathrm{c}$ that
    \[
    \Phi_{\tilde{G}_{t'}}(p_i) = \frac{w_{\tilde{G}_{t'}}(p_i, \tilde{V}_{t'} \setminus p_i)}{\min\{\vol_{\tilde{G}_{t'}}(p_i), \vol_{\tilde{G}_{t'}}(\tilde{V}_{t'} \setminus p_i)\}} =  O \lp k^{-6} \cdot \log^{-\gamma}(n_t) \rp,
    \]
    which proves the second part of statement~(C3).  
\end{proof}

\begin{corollary}\label{corr:vol_change_subset}
   Suppose for a subset $A \subset V_{t'}$ with $\vol_{G_{t'}}(A) \leq \vol({G_{t'}})/2$, it holds that $\Phi_{\tilde{G}_{t'}}(\tilde{A}) > (21 \cdot c_1)\cdot \log^{-\varepsilon}(n_{t'})$ and $\Phi_{G_{t'}}(A) \leq \log^{-\varepsilon}(n_{t'})$ for any positive  $c_1, \epsilon$ satisfying $4 \cdot c_1 \leq \log^{\varepsilon}(n_{t'})$. Then, it holds that 
    $
    \Phi_{G_t}(A^{(t)}) < 1/c_1.
    $
\end{corollary}

\begin{proof}[Proof of Corollary~\ref{corr:vol_change_subset}]
 Assume by contradiction that $\Phi_{G_t}(A^{(t)}) \geq 1/c_1 $. Then, by statement (C2) in Lemma~\ref{lem:new_cluster_properties} and the fact that $\Phi_{G_{t'}}(A) \leq \log^{-\varepsilon}(n_{t'})$, it holds that $\Phi_{\tilde{G}_{t'}}(\tilde{A}) \leq (21 \cdot c_1)\cdot \log^{-\varepsilon}(n_{t'})$, which is a contradiction. Hence, it holds that  $\Phi_{G_t}(A^{(t)}) < 1/c_1 $.
\end{proof}

 Before analysing the spectral gap in the contracted graph $\tilde{G}_{t'}$ with respect to the spectral gap in the full graph $G_{t'}$, we show that for any small subset of vertices $A \subset V$ with a low value of $\Phi_{G_{t'}}(A)$,   the conductance of its corresponding set in the contracted graph $\Phi_{\tilde{G}_{t'}}(\tilde{A})$ is  low as well.

\begin{lemma}\label{lem:low_conductance_full_vs_contract_graph}
    Let $C \subset V_{t'}$ be a subset of vertices such that $\vol_{G_{t'}}(C) \leq k^6 \cdot \log^{2\gamma}(n_t)$ and $\Phi_{G_{t'}} (C) \leq \log^{-\varepsilon}(n_{t'})$ for some   constant $\varepsilon>0$. Then, it holds that 
     $$\Phi_{\tilde{G}_{t'}} (\tilde{C}) = O(\log^{-0.9 \varepsilon}(n_{t'})).$$ 
\end{lemma}

\begin{proof} 
    We   prove this by contradiction. Assume by contradiction that 
    $$\Phi_{\tilde{G}_{t'}} (\tilde{C}) > \frac{21}{4} \cdot \log^{0.1 \varepsilon}(n_{t'}) \cdot \log^{-\varepsilon}(n_{t'}) = \frac{21}{4} \cdot \log^{-0.9 \varepsilon}(n_{t'}).$$
    Setting $c_1 \triangleq (1/4) \cdot \log^{0.1 \varepsilon}(n_{t'})$, it holds by Corollary~\ref{corr:vol_change_subset} that 
\begin{equation}\label{eq:conductance_C_contradiction}
        \Phi_{G_{t}} (C^{(t)}) < 4 \cdot \log^{-0.1\varepsilon}(n_{t'}).
    \end{equation}
  We will show that  $C^{(t)}$ can be used to create a $(k+1)$-partition in $G_t$ with low outer conductance, contradicting with the fact that $\lambda_{k+1}(G_t) = \Omega(1)$.

 Let $S_1, \ldots S_k$ be the optimal clusters in $G_t$ corresponding to $\rho_{G_t}(k)$. Given that $G_t$ is a connected graph and $\rho_{G_t}(k) =   O \lp k^{-8} \cdot \log^{-2\gamma}(n_t) \rp$, it holds that $\vol_{G_t}(S_i) = \Omega \lp k^{8} \cdot \log^{2\gamma}(n_t) \rp$ for all $1 \leq i \leq k$.    We then create the following $(k+1)$-partition:  
 \[
\mathcal{A} \triangleq C^{(t)} \cup \left\{S_{1} \setminus C^{(t)}, \ldots, S_{k} \setminus C^{(t)} \right\}, 
\]
which is a valid partition as we know that  $\vol_{G_t}(C^{(t)}) \leq \vol_{G_{t'}}(C) \leq k^6 \cdot \log^{2\gamma}(n_t)$ by  the conditions of the lemma.
Now we will compute the conductance of each cluster in $\mathcal{A}$. 

 First of all, we have  from \eqref{eq:conductance_C_contradiction}   that
\begin{equation}\label{eq:case1_conductance_C_in_contract}
\Phi_{G_{t}} (C^{(t)}) < 4 \cdot \log^{-0.1\varepsilon}(n_{t'}).
\end{equation}
Second, for any cluster $S_j \setminus C^{(t)}$  we have that
 \[
 \Phi_{G_t}(S_j \setminus C^{(t)}) = \frac{w_{G_t}(S_j \setminus C^{(t)}, V_t \setminus (S_j \setminus C^{(t)}))}{\min\{\vol_{G_t}(S_j \setminus C^{(t)}), \vol_{G_t}(V_t \setminus (S_j \setminus C^{(t)}))\}}.
 \]
 Our proof is by the following   case distinction:
 
 \textit{Case 1: $\min\{\vol_{G_t}(S_j \setminus C^{(t)}), \vol_{G_t}(V_t \setminus (S_j \setminus C^{(t)}))\} = \vol_{G_t}(V_t \setminus (S_j \setminus C^{(t)}))$}. 
\begin{align}
     \Phi_{G_t}(S_j \setminus C^{(t)}) &= \frac{w_{G_t}(S_j \setminus C^{(t)}, V_t \setminus (S_j \setminus C^{(t)}))}{\vol_{G_t}(V_t \setminus (S_j \setminus C^{(t)}))}  \nonumber \\
     &\leq \frac{w_{G_t}(S_j, V_t \setminus S_j) + w_{G_t}(C^{(t)}, V_t \setminus C^{(t)})}{\vol_{G_t}(V_t \setminus S_j) + \vol_{G_t}(C^{(t)} \cap S_j)}  \nonumber \\
     &\leq 2 \cdot \max \left\{\frac{w_{G_t}(S_j, V_t \setminus S_j) }{\vol_{G_t}(V_t \setminus S_j)},  \frac{ w_{G_t}(C^{(t)}, V_t \setminus C^{(t)})}{\vol_{G_t}(V_t \setminus S_j)}\right\} \nonumber \\
     &\leq 2 \cdot \max\left\{\Phi_{G_t}(S_j), \Phi_{G_t}(C^{(t)})\right\}  \label{large_cluster_easy}\\
     &\leq \max\{O \lp k^{-8} \cdot \log^{-2\gamma}(n_t) \rp, 4 \cdot \log^{-0.1\varepsilon }(n_{t'})\}, \nonumber
  \end{align}
 where for \eqref{large_cluster_easy} it holds that $\min\{\vol_{G_t}(S_j),\vol_{G_t}(V_t \setminus S_j) \} = \vol_{G_t}(V_t \setminus S_j)$ because we know that   $\vol(G_t)/2 \geq \vol_{G_t}(V_t \setminus (S_j \setminus C^{(t)})) \geq \vol_{G_t}(V_t \setminus S_j)$, and we also know that $\vol_{G_t}(V_t \setminus S_j) \geq \vol_{G_t}(C^{(t)})$.

 \textit{Case 2: $\min\{\vol_{G_t}(S_j \setminus C^{(t)}), \vol_{G_t}(V_t \setminus (S_j \setminus C^{(t)}))\} = \vol_{G_t}(S_j \setminus C^{(t)})$}.
 \begin{align}
     \Phi_{G_t}(S_j \setminus C^{(t)}) &= \frac{w_{G_t}(S_j \setminus C^{(t)}, V_t \setminus (S_j \setminus C^{(t)}))}{\vol_{G_t}(S_j \setminus C^{(t)})}  \nonumber \\
      &\leq \frac{w_{G_t}(S_j, V_t \setminus S_j) + w_{G_t}(C^{(t)}, V_t \setminus C^{(t)})}{\vol_{G_t}(S_j) - \vol_{G_t}(C^{(t)})}  \nonumber \\
      &\leq \frac{w_{G_t}(S_j, V_t \setminus S_j) + w_{G_t}(C^{(t)}, V_t \setminus C^{(t)})}{\Omega \lp \vol_{G_t}(S_j) \rp} \label{large_cluster_c}\\
      &= O \lp \Phi_{G_t}(S_{j}) \rp + O \lp \Phi_{G_t}(C^{(t)}) \rp \label{large_cluster_i} \\
      &=  2 \cdot \max \left\{O \lp k^{-8} \cdot \log^{-2\gamma}(n_t) \rp, O \lp\log^{-0.1\varepsilon}(n_{t'})\rp\right\}
 \end{align}
 where \eqref{large_cluster_c} holds because  $\vol_{G_t}(S_i) = \Omega \lp k^{8} \cdot \log^{2\gamma}(n_t)  \rp$ and $\vol_{G_t}(C^{(t)}) =  O \lp k^{6} \cdot \log^{2\gamma}(n_t)  \rp$, 
\eqref{large_cluster_i} holds because $w_{G_t}(S_j, V_t \setminus S_j) \leq \Phi_{G_t}(S_j) \cdot \vol_{G_t}(S_j)$ and $w_{G_t}(C^{(t)}, V_t \setminus C^{(t)}) \leq \Phi_{G_t}(C^{(t)}) \cdot \vol_{G_t}(C^{(t)})$.  

 Combining both cases, we have for every $1 \leq j \leq k$ that
\begin{equation}\label{eq:case2_conductance_C_in_contract}
\Phi_{G_t}(S_j \setminus C^{(t)}) =  2 \cdot \max \left\{O \lp k^{-8} \cdot \log^{-2\gamma}(n_t) \rp, O \lp\log^{-0.1\varepsilon}(n_{t'})\rp\right\}.
\end{equation}
 Therefore, by combining \eqref{eq:case1_conductance_C_in_contract} and~\eqref{eq:case2_conductance_C_in_contract}, we have shown that 
$$
\rho_{G_t}(k+1) \leq \max_{A_j \in \mathcal{A}} \Phi_{G_t}(A_j) = 2 \cdot \max \left\{O \lp k^{-8} \cdot \log^{-2\gamma}(n_t) \rp, O \lp\log^{-0.1\varepsilon}(n_{t'})\rp\right\},
$$
which contradicts the fact that $\rho_{G_t}(k+1) \geq \frac{\lambda_{k+1}(\mathcal{L}_{G_t})}{2} = \Omega(1)$. Hence, the statement of the lemma follows.
\end{proof}

\begin{proof}[Proof of Lemma~\ref{lem:conductance_full>contract}]

We first prove the first statement. 
        Let $\mathcal{S} = S_1, \ldots, S_\ell$ be a set of clusters that achieve $\rho_{G_{t'}}(\ell)$. For ease of notation we set
         \[
        \mathcal{S}_\mathrm{small} \triangleq \mathcal{S}^{(t')}_\mathrm{small}\lp  k^6 \cdot \log^{2\gamma}(n_t)\rp 
        \]
        to be the clusters in $\mathcal{S}$ with volume at most $k^6\cdot\log^{2\gamma}(n_t)$, and similarly
         \[
        \mathcal{S}_\mathrm{large} \triangleq \mathcal{S}^{(t')}_\mathrm{large}\lp k^6 \cdot \log^{2\gamma}(n_t) \rp.
        \]
     We will   use the partition $\mathcal{S}$, which has low outer conductance in $G_{t'}$, to create an $r$-way partition in $\tilde{G}_{t'}$ with low $r$-way expansion. We construct this $r$-way partition, denoted by $\mathcal{R}$, as follows:  
    \[
    \mathcal{R} \triangleq \left\{\widetilde{S}_1, \ldots, \widetilde{S}_{\ell_1}, p_1, \ldots, p_{k-1}, p^*_k  \right\}
    \]
    where $\ell_1 \triangleq |\mathcal{S}_\mathrm{small}|$, and we define
    $$
    p^*_k \triangleq p_k \cup \left(\tilde{V}_{t'}^\mathrm{nc} \setminus \bigcup_{S_j \in \mathcal{S}_\mathrm{small}} \tilde{S}_j\right)
    $$ to be the union of the super node $p_k$ with the leftover non-contracted vertices which do not belong to any $\tilde{S}_j$. We start by showing that $\mathcal{R}$ has low $r$-way expansion: 
    \begin{itemize}
        \item  By Lemma~\ref{lem:low_conductance_full_vs_contract_graph}, we know that for every $S_j \in \mathcal{S}_\mathrm{small}$, it holds that $\Phi_{\tilde{G}_{t'}}(\tilde{S}_j) = O \lp   \log^{-0.9\alpha}(n_{t'}) \rp$.
        \item  By Property~(C3) of Lemma~\ref{lem:new_cluster_properties}, we know that for every super node $p_i \in \{p_{k_1}, \ldots p_{k-1}\}$ it holds that $ \Phi_{\tilde{G}_{t'}}(p_i) = O \lp k^{-6} \cdot \log^{-\gamma}(n_{t'}) \rp$.
        \item  Finally, for $p_k^*$ we know that
        \begin{align}
            \Phi_{\tilde{G}_{t'}}(p^*_k) = \frac{w_{\tilde{G}_{t'}}(p^*_k, \tilde{V}_{t'} \setminus p^*_k)}{\min\left\{\vol_{\tilde{G}_{t'}}(p^*_k), \vol_{\tilde{G}_{t'}}(\tilde{V}_{t'} \setminus p^*_k) \right\}}.
        \end{align}
         We split the computation of this conductance into two cases.
        
        \textit{Case~1:}  Suppose $\min\left\{\vol_{\tilde{G}_{t'}}(p^*_k), \vol_{\tilde{G}_{t'}}(\tilde{V}_{t'} \setminus p^*_k) \right\} = \vol_{\tilde{G}_{t'}}(p^*_k)$. Then, we have that 
       \begin{align}
            \Phi_{\tilde{G}_{t'}}(p^*_k) &= \frac{w_{\tilde{G}_{t'}}(p^*_k, \tilde{V}_{t'} \setminus p^*_k)}{\vol_{\tilde{G}_{t'}}(p^*_k)} \nonumber\\
            &\leq \frac{w_{\tilde{G}_{t}}(p_k, \tilde{V}_t \setminus p_k) +  \log^{\gamma}(n_{t'}) + \vol_{\tilde{G}_{t'}}(\tilde{V}_{t'}^\mathrm{nc}) }{\vol_{\tilde{G}_{t}}(p_k) - \vol_{\tilde{G}_{t'}}(\tilde{V}_{t'}^\mathrm{nc}) } \label{eq1:special_cluster_conductance_case1}\\
             &\leq \frac{\Phi_{\tilde{G}_t}(p_k) \cdot \vol_{\tilde{G}_{t}}(p_k) +  3 \cdot \log^{\gamma}(n_{t'}) }{\Omega \lp \vol_{\tilde{G}_{t}}(p_k) \rp} \label{eq2:special_cluster_conductance_case1}\\
             &= \frac{ O \lp k^{-6} \cdot \log^{-2\gamma}(n_{t}) \rp \cdot \vol_{\tilde{G}_{t'}}(p_k) +  3 \cdot \log^{\gamma}(n_{t'}) }{\Omega \lp \vol_{\tilde{G}_{t}}(p_k) \rp} \label{eq3:special_cluster_conductance_case1} \\
             &= \frac{O \lp k^{-6} \cdot \log^{-2\gamma}(n_{t})   \cdot \vol_{\tilde{G}_{t'}}(p_k) \rp}{\Omega \lp \vol_{\tilde{G}_{t}}(p_k) \rp} \label{eq4:special_cluster_conductance_case1} \\
             & = O \lp k^{-6} \cdot \log^{-2\gamma}(n_{t}) \rp \nonumber,
        \end{align}
        where \eqref{eq1:special_cluster_conductance_case1} holds because $|\Enew| \leq \log^{\gamma}(n_{t})$ is the maximum amount of weight that can be added between $p_k$ and its complement, \eqref{eq2:special_cluster_conductance_case1} holds because  $\vol_{\tilde{G}_{t'}}(\tilde{V}_{t'}^\mathrm{nc}) \leq 2 \cdot |\Enew|$ is the  maximum volume of non-contracted vertices that can be added to $\tilde{G}_{t'}$  and \eqref{eq3:special_cluster_conductance_case1} holds because of statement~(C3) of Lemma~\ref{lem:new_cluster_properties}, and \eqref{eq4:special_cluster_conductance_case1} holds
        since $\vol_{\tilde{G}_{t'}}(p_k) \geq \vol(G_t)/k \geq n_t/k$. 
        
        \textit{Case~2:} Suppose $\min\left\{\vol_{\tilde{G}_{t'}}(p^*_k), \vol_{\tilde{G}_{t'}}(\tilde{V}_{t'} \setminus p^*_k) \right\} = \vol_{\tilde{G}_{t'}}(\tilde{V}_{t'} \setminus p^*_k)$. Then it holds that the conductance of $p^*_k$ is upper bounded by the maximum conductance of every other cluster in $\mathcal{R}$, i.e.,
        \begin{align*}
            \Phi_{\tilde{G}_{t'}}(p^*_k) &= \frac{w_{\tilde{G}_{t'}}(p^*_k, \tilde{V}_{t'} \setminus p^*_k)}{\vol_{\tilde{G}_{t'}}(\tilde{V}_{t'} \setminus p^*_k)} \nonumber\\
            &\leq \frac{\sum_{S_j \in \mathcal{S}_\mathrm{small}} w_{\tilde{G}_{t'}}(\tilde{S}_j, \tilde{V}_{t'} \setminus \tilde{S}_j) + \sum_{j=1}^{k-1} w_{\tilde{G}_{t'}}(p_j, \tilde{V}_{t'} \setminus p_j)}{\sum_{S_j \in \mathcal{S}_\mathrm{small}} \vol_{\tilde{G}_{t'}}(\tilde{S}_j) + \sum_{j=1}^{k-1} \vol_{\tilde{G}_{t'}}(p_j)} \nonumber\\
            & \leq \max \left\{ \max_{S_j \in \mathcal{S}_\mathrm{small}} \left\{ \Phi_{\tilde{G}_{t'}}(\widetilde{S}_j) \right\}, \max_{p_j \in \{p_{1}, \ldots, p_{k-1}\}} \left\{ \Phi_{\tilde{G}_{t'}}(p_j) \right\}\right\} \\
            & = \max\left\{O \lp  \log^{-0.9\alpha}(n_{t'}) \rp,  O \lp k^{-6} \cdot \log^{-\gamma}(n_{t'}) \rp \right\},
        \end{align*}
        where the last inequality follows by the mediant inequality. 
    \end{itemize}

    Combining the two cases, we have that 
        \[
        \Phi_{\tilde{G}_{t'}}(p^*_k) = \max\left\{O \lp  \log^{-0.9\alpha}(n_{t'}) \rp,  O \lp k^{-6} \cdot \log^{-\gamma}(n_{t'}) \rp \right\}.
        \]     
    We have so far analysed   the conductance of each of the clusters in the partition $\mathcal{R}$, and have   shown that
    \begin{equation}\label{eq:r_conductance}
    \rho_{\tilde{G}_{t'}}(r) =  \max\left\{O \lp  \log^{-0.9\alpha}(n_{t'}) \rp,  O \lp k^{-6} \cdot \log^{-\gamma}(n_{t'}) \rp \right\}.
    \end{equation}
    Before reaching the final contradiction, we prove the following claim.
    \begin{claim}\label{claim:r_vs_ell}
        It holds that $r \geq \ell$.  
    \end{claim}
    \begin{proof}
          Assume by contradiction that 
    $r < \ell$.
    In this case, we know   that $r = |\mathcal{S}_\mathrm{small}| + |\mathcal{P}|$, and $\ell = |\mathcal{S}|$. Therefore, the condition of $r<\ell$ gives us that $|\mathcal{S}_\mathrm{small}| + |\mathcal{P}| < |\mathcal{S}|$, which implies that $|P| < |\mathcal{S}| - |\mathcal{S}_\mathrm{small}| = |\mathcal{S}_\mathrm{large}|$. This means that the number of large clusters in $\mathcal{S}$ is greater than the number of clusters in $\mathcal{P}$.

    It therefore holds that $|\mathcal{S}_\mathrm{large}| > |\mathcal{P}| = k$. Furthermore, since it holds that for every $S_j \in \mathcal{S}_\mathrm{large}$ that $\vol_{G_{t'}}(S_j) >  k^6 \cdot \log^{2\gamma}(n_t)$, and the number of new edges is $|\Enew| \leq  \log^{\gamma}(n_t)$, it also holds that
    \[
    \Phi_{G_{t}}(S_j) = O\lp \Phi_{G_{t'}}(S_j) \rp = \max \left\{ O \lp \log^{- \alpha}(n_{t'}) \rp,  k^6 \cdot \log^{\gamma}(n_t)\right\}.
    \]    
    This means that $\mathcal{S}_\mathrm{large}$ is a set of $|\mathcal{S}_\mathrm{large}| \geq k+1$ disjoint subsets in $G_t$ with low conductance, which contradicts the higher-order Cheeger inequality and proves the claim.
    \end{proof}
    Combining \eqref{eq:r_conductance}  with Claim~\ref{claim:r_vs_ell} gives us that 
    \[
    \rho_{\tilde{G}_{t'}}(\ell) = \max\left\{O \lp  \log^{-0.9\alpha}(n_{t'}) \rp,  O \lp k^{-6} \cdot \log^{-\gamma}(n_{t'}) \rp \right\},
    \]
    and this proves the first statement.

    Next we prove the second statement. 
    Let $A_1, \ldots A_\ell$ be the partition such that  $\Phi_{\tilde{G}_{t'}}(A_i) = O\lp \log^{-\delta}(n_{t'})\rp$ for every $1 \leq i \leq \ell$.  Recall that $\hat{A_i}$ is the representation of the set $A_i$ in the full graph $G_{t'}$, i.e.,
     $$
    \hat{A_i} \triangleq A_i^\mathrm{nc} \bigcup \left(\bigcup_{p_j \in A_i^\mathrm{c}}P^{(t')}_j\right),
    $$
     where $P^{(t')}_j = P_j \setminus (P_j \cap \tilde{V}^{\mathrm{nc}}_{t'})$, $A_i^\mathrm{nc} \triangleq A_i \cap V^{\mathrm{nc}}_{t'}$, and $A_i^\mathrm{c} \triangleq A_i \cap V^{\mathrm{c}}_{t'}$. One can see $P^{(t')}_j$ as the vertices in $P_j$ that have not been pulled out into the contracted graph yet.

    Notice that, when  $A_i^\mathrm{c} = \emptyset$,   it holds by construction that $\Phi_{G_{t'}}(A_i) 
 = \Phi_{\tilde{G}_{t'}}(A_i) \leq \log^{-\delta}(n_{t'})$.
  So we only look at the case where $A^\mathrm{c}_i \neq \emptyset $. Without loss of generality, we assume that $A^c_i$ does not contain all the contracted nodes $p_1, \ldots, p_k$. If it did, then $$\Phi_{G_{t'}}(\hat{A}_i) = \Phi_{G_{t'}}\lp \bigcup_{A_j^c = \emptyset} A_j \rp \leq \log^{-\delta}(n_{t'}).$$ Therefore, given that for any $1\leq i \leq \ell$ it holds that $\vol_{G_{t'}}(A_i^\mathrm{nc}) \leq 2 \cdot |\Enew| \leq 2 \cdot \log^{\gamma}(n_t)$, and for any $1 \leq j \leq k$ it holds that $\vol_{G_{t'}}(P_i^{(t')}) = \Omega(k^6 \cdot \log^{2\gamma}(n_t))$, we get that 
 \[
 \Phi_{G_{t'}}(\hat{A}_i) = O \lp \Phi_{G_{t'}} \lp \bigcup_{p_j \in A_i^\mathrm{c}}P^{(t')}_j \rp \rp = O\lp k^{-6} \cdot \log^{-\gamma}(n_{t'})\rp,
 \]
 where the last line holds because of property (C3) of Lemma~\ref{lem:new_cluster_properties}. This proves the second statement.
\end{proof}

\begin{proof}[Proof of Lemma~\ref{lem:lambda_k+1 from tilde(G) to G}]
We first prove the first statement, and we will prove this by contradiction. Assume by contradiction that $\lambda_{\ell+1}(\mathcal{L}_{G_{t'}}) < C\cdot \frac{\log^{-\alpha}(n_{t'})}{(\ell+1)^6}$ for some constant $C$. Then, by the higher-order Cheeger inequality, there exists an optimal $(\ell+1)$-way partition $\mathcal{S} = \{S_1, \ldots S_{\ell+1}$\} such that for all $1 \leq i \leq \ell+1$
\begin{equation*}\label{eq:conductance_contradiction_sets}
      \Phi_{G_{t'}}(S_i) \leq \rho_{G_{t'}}(\ell+1) \leq C_{\ref{lem:Higher Cheeger}} \cdot (\ell+1)^3 \cdot \sqrt{\lambda_{\ell+1}(G_{t'})} =O\left( \log^{-0.5\alpha}(n_{t'})\right).  
    \end{equation*}
    By Lemma~\ref{lem:conductance_full>contract}, it then holds that  $\rho_{\tilde{G}_{t'}}(\ell+1) = \max\left\{O \lp  \log^{-0.45\alpha}(n_{t'}) \rp, O\lp k^{-6} \cdot \log^{-\gamma}(n_{t'}) \rp\right\}$, which contradicts the fact that 
    $$\rho_{\tilde{G}_{t'}}(\ell+1) \geq \frac{\lambda_{\ell+1}(\mathcal{L}_{\tilde{G}_{t'}})}{2} = \Omega(1),$$
    from which the  first statement of the lemma follows.

    Next we prove the second statement.  
 We   prove this by analysing the spectrum of $\mathcal{L}_{\tilde{G}_{t'}}$ with respect to $\mathcal{L}_{G_{t'}}$ through $\mathcal{L}_{H'_{t'}}$. As proven in Lemma~\ref{lem:special_sparsifier}, $H'_{t'}$ is a cluster preserving sparsifier of $G_{t'}$, and therefore we know that 
\begin{equation}\label{eq:lowerbound_ell_SZ19}
    \lambda_{\ell+1}(\mathcal{L}_{H'_{t'}}) = \Omega(\lambda_{\ell+1}(\mathcal{L}_{G_{t'}})).
\end{equation}

Our next analysis is inspired by the work on meta graphs of Macgregor and Sun~\cite{MS22}. We will analyse the spectrum of $\mathcal{L}_{H'_{t'}}$ with respect to the spectrum of $\mathcal{L}_{\tilde{G}_{t'}}$, and for simplicity we denote $H'_{t'} \triangleq H$,  $\tilde{G}_{t'} \triangleq \tilde{G}$, and $n_{t'} \triangleq n$. 
 For every vertex $u_j \in V(\tilde{G})$ in the contracted graph, we associate it with a non-empty group of vertices $A_j \subset V(H)$ as follows: for all $u_j \in \tilde{V}_{t'}^\mathrm{nc}$, we associate $u_j$ with its unique corresponding single vertex $v\in V(H)$, and for every $u_j = p_r \in \tilde{V}_{t'}^\mathrm{c}$ for some $r$, we associate it with its corresponding vertices in the cluster $P^{(t')}_r \subset V(H)$. Then, let $\chi_j \in \mathbb{R}^n$ be the indicator vector for the vertices $A_j \subset V(H)$ corresponding to the vertex $u_j \in V(\tilde{G})$. 

 We define $\tilde{n} = |V(\tilde{G})|$, and let  the eigenvalues of $\mathcal{L}_{\tilde{G}}$ be $\gamma_1 \leq \gamma_2 \leq \ldots \leq \gamma_{\tilde{n}}$ with corresponding eigenvectors $g_1, g_2, \ldots, g_{\tilde{n}} \in \mathbb{R}^{\tilde{n}}$. We further define vectors $\bar{g}_1, \ldots \bar{g}_{\tilde{n}}$ which will represent the eigenvectors $g_1, \ldots g_{\tilde{n}}$ of the normalised Laplacian $\mathcal{L}_{\tilde{G}}$, but blown up to size $\mathbb{R}^n$. Formally,  we define
\[
\bar{g}_i \triangleq \sum_{j=1}^{\tilde{n}} \frac{D_H^{\frac{1}{2}}\chi_j}{\|D_H^{\frac{1}{2}}\chi_j\|} g_i(j).
\]
 We can readily check that these vectors form an orthonormal basis. First, 
\begin{align*}
    \bar{g}_i \bar{g}_i^\intercal &= \sum_{j=1}^{\tilde{n}}  \sum_{u \in A_j} \lp \frac{\sqrt{d_H(u)}}{\sqrt{\vol_H(A_j)}} g_i(j) \rp^2 \\
    &=\sum_{j=1}^{\tilde{n}} g_i(j)^2 \sum_{u \in A_j}  \frac{d_H(u)}{\vol_H(A_j)}  \\
    &=\sum_{j=1}^{\tilde{n}} g(j)^2 = 1.
\end{align*}
And similarly for any $i_1 \neq i_2$,
\begin{align*}
\bar{g}_{i_1} \bar{g}_{i_2}^\intercal &= \sum_{j=1}^{\tilde{n}} \sum_{u \in A_j} \frac{d_H(u)}{\vol_H(A_j)} g_{i_1}(j) g_{i_2}(j) \\
&=\sum_{j=1}^{\tilde{n}} g_{i_1}(j) g_{i_2}(j) = 0.
\end{align*}
We also get the useful property that for the eigenvalues $\lambda_1, \ldots, \lambda_{n}$ of $\mathcal{L}_H$ and $\gamma_1, \ldots, \gamma_{\tilde{n}}$ of the contracted Laplacian $\mathcal{L}_{\tilde{G}}$, it holds that $\lambda_i \leq 2 \cdot \nu_i$. In particular,
\begin{align*}
    \bar{g}_i \mathcal{L}_H \bar{g}_i^\intercal &= \sum_{x=1}^{\tilde{n}} \sum_{y=x}^{\tilde{n}} \sum_{u \in A_x} \sum_{v \in A_y} w_H(u,v)\lp \frac{\bar{g}_i(u)}{\sqrt{d_H(u)}} - \frac{\bar{g}_i(v)}{\sqrt{d_H(v)}}\rp^2 \\
    &= \sum_{x=1}^{\tilde{n}} \sum_{y=x}^{\tilde{n}}  w_H(A_{x},A_{y})\lp \frac{g_i(x)}{\sqrt{\vol_H(A_{x}})} - \frac{g_i(y)}{\sqrt{\vol_H(A_{y}})}\rp^2 \\
    &= 2 \cdot g_i \mathcal{L}_{\tilde{G}} g_i^\intercal.
\end{align*}
Therefore we have an $i$-dimensional subspace $X_i$ such that 
\[
\max_{x\in X_i} \frac{x^{\intercal}\mathcal{L}_{H} x}{x^\intercal x} = 2 \cdot \gamma_i,
\]
from which it follows by the Courant-Fischer theorem that $\lambda_i \leq 2 \cdot \gamma_i$. Combining this with~\eqref{eq:lowerbound_ell_SZ19}, we get that  
$$
\lambda_{\ell+1}(\mathcal{L}_{\tilde{G}_{t'}}) \geq \frac{1}{2}\cdot\lambda_{\ell+1}(\mathcal{L}_{H'_{t'}}) = \Omega \lp \lambda_{\ell+1}(\mathcal{L}_{G_{t'}}) \rp = \Omega(1),
$$
which proves the lemma.
\end{proof}

\end{document}